\documentclass[11pt,letterpaper]{scrartcl}
\usepackage[utf8]{inputenc}
\usepackage[margin=1in]{geometry}

\usepackage{times}
\usepackage{amsmath}
\usepackage{amsfonts}
\usepackage{amsthm}

\usepackage{multirow}
\usepackage[table]{xcolor}

\usepackage{authblk}

\newtheorem{lemma}{Lemma}
\newtheorem{conjecture}{Conjecture}
\newtheorem{theorem}{Theorem}
\newtheorem{definition}{Definition}

\newtheorem{corollary}{Corollary}
\newtheorem{fact}{Fact}

\usepackage[maxbibnames=9,maxcitenames=2,backend=bibtex]{biblatex}
\usepackage{subcaption}

\usepackage{tikz}
\usetikzlibrary{fit,shapes.misc,shapes.geometric,decorations.pathreplacing}
\pgfdeclarelayer{bg}
\pgfsetlayers{bg,main}

\title{Independent Sets in Vertex-Arrival Streams}

\author[1]{Graham Cormode}
\author[1]{Jacques Dark}

\author[2]{Christian Konrad}
\affil[1]{Department of Computer Science, University of Warwick, Coventry, UK\\ \texttt{\{g.cormode,j.dark\}@warwick.ac.uk}}
\affil[2]{Department of Computer Science, University of Bristol, Bristol, UK\\ \texttt{christian.konrad@bristol.ac.uk}}

\date{}

\begin{document}

\maketitle

\begin{abstract}
    We consider the classic maximal and maximum independent set problems in
    three models of graph streams:

    In the edge-arrival model we see a stream of edges which
    collectively define a graph; this model has been well-studied for
    a variety of problems. We first show that the space complexity for
    a one-pass streaming algorithm to find a maximal independent set
    is quadratic (i.e. we must store all edges). We further
    show that the problem does not become much easier if we only
    require approximate maximality. This contrasts strongly with the
    other two vertex-based models, where one can greedily find an exact solution using only the space needed to store the independent set.

    In the ``explicit'' vertex stream model, the input stream is a
    sequence of vertices making up the graph, where every vertex
    arrives along with its incident edges that connect to previously
    arrived vertices. Various graph problems require substantially
    less space to solve in this setting than for edge-arrival streams. We
    show that every one-pass c-approximation streaming algorithm for
    maximum independent set (MIS) on explicit vertex streams requires space $\Omega(\frac{n^2}{c^7})$, 
    where $n$ is the number of vertices of the input graph, and it is
    already known that space $\tilde{\Theta}(\frac{n^2}{c^2})$ 
    is necessary and sufficient in the edge arrival model (Halldórsson et al. 2012). The MIS problem is thus not significantly 
    easier to solve under the explicit vertex arrival order assumption. Our result is proved via a reduction to a new 
    multi-party communication problem closely related to pointer jumping.
    
    In the ``implicit'' vertex stream model, the input stream consists of a sequence of objects, one per vertex. 
    The algorithm is equipped with a function that can map a pair of objects to the presence or absence of an edge, 
    thus defining the graph. This model captures, for example, geometric intersection graphs such as unit disc graphs. 
    Our final set of results consists of several improved upper and lower bounds for ball intersection graphs, in both 
    explicit and implicit streams.

\end{abstract}
\thispagestyle{empty}
\addtocounter{page}{-1}
\clearpage
\section{Introduction}

The streaming model supposes that, rather than being loaded into memory all at once, the input is received piece-by-piece over a period of time. Only a sublinear amount of memory (in the input size) is made available, preventing any algorithm from ``seeing'' even a constant fraction of the whole input at once.

In graph streams (see \cite{m14} for an excellent survey), we distinguish between the ``edge-arrival'' model, where the stream consists of individual edges arriving in any order, and the ``vertex-arrival'' model, where the stream consists of batches of edges incident to a particular vertex---as each vertex ``arrives'' we are given all the edges from the new vertex to previously arrived vertices. We will shorten the names to edge streams and vertex streams, respectively. Problems on edge streams are always at least as hard as on vertex streams (as any vertex stream is also a valid edge stream).

There is a further variant which we will call ``implicit'' vertex streams (as opposed to the normal explicit representation). In this model, the stream consists of a series of small ($\text{polylog}(n)$-sized) identifiers, one per vertex. We are additionally provided with some function or oracle which maps a pair of identifiers to a Boolean output indicating whether the two vertices are connected or not. This implicitly defines a graph over the list of identifiers received.
Geometric intersection graphs, received as a stream of geometric objects, are the most natural member of this class. For example, a unit interval intersection graph might be received as a series of points in $\mathbb{R}$. Then a pair of vertices $x, y$ are adjacent if and only if $|x - y| \leq 1$.

Explicit and implicit vertex streams are closely related but distinct,
with neither being strictly ``harder'' than the other. For example: it
is easy to count exactly the number of edges in $O(1)$ space (of
words) for an explicit vertex stream, however, doing so for an
implicit stream requires linear space, otherwise we cannot hope to
know how many edges are incident to the final vertex.
On the other hand: implicit vertex streams can be stored entirely in $\tilde{O}(n)$ space\footnote{We use $\tilde{O}$, $\tilde{\Theta}$, and $\tilde{\Omega}$ to mean $O$, $\Theta$, and $\Omega$ (respectively) with log factors suppressed.}, whereas explicit vertex streams require $\Omega(n^2)$ space to store the full structure.

Maximum independent set is an important problem on graphs. The task is to find a largest subset of vertices which have no edges between them. Unfortunately, the offline problem is NP-hard to find a maximum set in a general graph \cite{k72}, and even hard to approximate within a factor of $n^{1-\epsilon}$, for any $\epsilon > 0$ \cite{z07}. It is also known to be hard (requiring $\tilde{\Theta}(\frac{n^2}{c^2})$ space to $c$-approximate on an $n$-vertex graph) in the edge streaming model, despite being allowed unlimited computation \cite{halldorsson2012streaming}. However, we can do much better for graphs of bounded independence, given as vertex streams \cite{Cormode:Dark:Konrad:18}.

In this paper, we study the hardness of approximate maximum independent set in the explicit and implicit vertex streaming model.
First, we motivate the study of vertex streams by showing a large space complexity gap between the edge and (both) vertex models for the related problem of finding a \emph{maximal} independent set. 
The gap remains non-trivial even when we only ask for a set which is ``nearly'' maximal.
Then, we propose a new communication problem (closely related to index and pointer jumping) and use it to show that MIS for general 
graphs cannot be much easier in the explicit vertex streaming model than in the edge streaming model.
Last, we show various improved upper and lower bounds for certain geometric intersection graphs in both vertex streaming models.

\subsection{Our Contributions}
For a graph $G$, let $\alpha(G)$ denote the size of a maximum independent set. Further, let $\chi(G)$ 
denote the {\em chromatic number}, i.e., the number of colors needed in any (legal) coloring of the input graph.
We arrange our contributions according to the three models of
streaming graphs that we consider (defined more formally below):

\begin{figure}[t]
    \begin{center}
        \begin{tabular}{|c||c|c|c|}\hline
             & \textbf{Approx. MIS} & \multicolumn{2}{|c|}{\textbf{Approx.} $\mathbf{\alpha(G)}$} \\
            \textbf{Space Bound} & $\tilde{O}(\alpha(G))$ & $\text{poly}(\log{n})$ & $\Omega(n)$ \\\hline\hline
            Unit Interval & $2$ (Greedy alg.) & $O\left(\frac{\log^2{n}}{\log\log{n}} \right)$ \cite{Cormode:Dark:Konrad:18} & \cellcolor{lightgray} $< 5/3$ \\\hline
        \end{tabular}
    \end{center}
    \caption{Approximation factors for Explicit Vertex Streams.}\label{fig:explicit-table}
\end{figure}

\begin{itemize}

\item
\textbf{Edge streams.}
We show that solving maximal independent set in an edge stream requires $\Omega(n^2)$ space, significantly more than the $\tilde{O}(\alpha(G))$ sufficient to solve it on an explicit or implicit vertex stream. We further show that even finding any independent set which covers all but an $n^{-\epsilon}$ fraction of the vertices for $\epsilon > 0$ (a nearly maximal independent set) in an edge stream requires $n^{2-o(1)}$ space, while covering just a constant fraction of the vertices requires $n^{1+\Omega(\frac{1}{\log\log{n}})}$ space.
The techniques rely on reductions from problems in communication
complexity --- initially, a relatively simple reduction from
\textsc{Index} to show the hardness of MIS, but a more involved
reduction to a graph-based variant of \textsc{Index} based on
Ruzsa-Szemer\'{e}di graphs for the approximate relaxation.

\item
\textbf{Explicit vertex streams.}
For general explicit vertex streams, we show that $c$-approximating
$\alpha(G)$ requires $\Omega(\frac{n^2}{c^7})$ space. We further observe 
that the construction leads to the same space bound for approximating $\chi(G)$, the chromatic number.
We rely on the communication complexity of a novel multiparty generalization of
the well-known \textsc{Index} problem, which is closely related to
pointer-jumping problems but allows us to more directly show
hardness. 
The hardness of the MIS problem itself then relies on a careful
construction based on erasure codes to bound the size of cliques in
our hard instance graphs.

\item
\textbf{Implicit vertex streams.}
Next, we show several results for geometric intersection graphs. We
can $3$-approximate MIS for a stream of unit squares in the plane
using $\tilde{O}(\alpha(G))$ space, and achieving better than a
$\frac{5}{2}$-approximation to $\alpha(G)$ requires $\Omega(n)$
space. Unit interval intersection graphs given as an explicit vertex
stream require $\Omega(n)$ space to get a better than $\frac{5}{3}$
approximation to $\alpha(G)$, making it harder than the implicit
stream equivalent. Figures~\ref{fig:implicit-table} and \ref{fig:explicit-table} illustrate these results
and put them in context to previously known bounds (see also Section~\ref{sec:prior-work}).
The two-dimensional upper bounds can be viewed as a generalization of
the one-dimensional bounds, with more work to cover the increased
number of cases that occur in 2D.
However, the lower bounds involve an intricate packing argument, to
show that the MIS size can be used to recover encoded information,
which is used in conjunction with our multiparty \textsc{Index}
variant problem to demonstrate approximation hardness.

\end{itemize}
\begin{figure}
    \begin{center}
        \begin{tabular}{|c||c|c|c|}\hline
             & \textbf{Approx. MIS} & \multicolumn{2}{|c|}{\textbf{Approx.} $\mathbf{\alpha(G)}$} \\
            \textbf{Space Bound} & $\tilde{O}(\alpha(G))$ & $\text{poly}(\log{n}, \epsilon^{-1})$ & $\Omega(n)$ \\\hline\hline
            Unit Interval & $3/2$ \cite{emek2016space}  & $3/2 + \epsilon$ \cite{cabello2017interval} & $< 3/2$ \cite{emek2016space} \\
            Interval & $2$ \cite{emek2016space} & $2 + \epsilon$ \cite{cabello2017interval} & $< 2$ \cite{emek2016space} \\
            Unit Square & \cellcolor{lightgray} $3$ & \cellcolor{lightgray} $3 + \epsilon$ & \cellcolor{lightgray} $< 5/2$ \\\hline
        \end{tabular}
    \end{center}
    \caption{Approximation factors for Implicit Vertex Streams. The first column concerns algorithms that output independent sets themselves, while the 
    second columns concerns algorithms that output estimations of the maximum independent set size. Our results are highlighted in gray.}\label{fig:implicit-table}
\end{figure}


\subsection{Problem Definitions}

Our problems are defined with respect to a graph $G = (V, E)$ with $n$ vertices ($|V| = n$).

\begin{definition}[Independent Sets]\leavevmode
An {\em independent set} of $G$ is a subset of the vertices $U
\subseteq V$ that contains no edges ($(U \times U) \cap E =
\emptyset$). 
We use $\alpha(G)$ to refer to the maximum size of any independent set
of $G$.
An MIS (maximum independent set) of $G$ is any independent set of size $\alpha(G)$. A maximal independent set is an independent set that no other independent sets contain as a proper subset.
\end{definition}

We are also interested in allowing approximations:

\begin{definition}[Approximate MIS]\leavevmode
  A $c$-approximate MIS of $G$ is an independent set of size at least $\frac{\alpha(G)}{c}$.
An algorithm is said to return a $c$-approximation to $\alpha(G)$ if it returns a value $\gamma$ satisfying $\frac{1}{c}\alpha(G) \leq \gamma \leq \alpha(G)$.
\end{definition}

The one-sided nature of the $c$-approximation requirement for
$\alpha(G)$ follows by analogy with the problem of finding a
$c$-approximate MIS, where of course no overestimate is possible.
A two-sided approximation can of course be made one-sided by 
rescaling by $\sqrt{c}$. 


\begin{definition}[Graph Streams]
We define the three different input arrival models as follows: 
\leavevmode
  \begin{itemize}
    \item An \textbf{edge stream} consists of a sequence of $m$ edges $S = \langle e_1, e_2, \cdots, e_m \rangle$ over a pre-determined vertex set $V$, arriving one-by-one in arbitrary order. Then the edge set is given by $E = \lbrace e_i : i \in [m] \rbrace$.

    \item An \textbf{explicit vertex stream} is a sequence of $n$
     vertices $S = \langle v_1, v_2, \cdots, v_n \rangle$
     arriving one-by-one in arbitrary order.
    Each vertex $v_i$ arrives with a list of all the edges between $v_i$
    and previous vertices ($v_j$ with $j < i$).
    Thus, each edge $\{v_i, v_j\}$ is presented exactly once in the
    stream, with the arrival of vertex $v_{\max \{i,j \}}$. 
   \item An \textbf{implicit vertex stream} consists of a sequence of
     $n$ identifiers $S = \langle I_1, I_2, \cdots, I_n \rangle$. Each
     identifier represents a vertex of the graph and is taken from
     some universe $[\mathcal{U}]$ with $\mathcal{U} \in
     O(2^{\operatorname{polylog}(n)})$.
     Thus, the identifiers have a succinct representation of $\operatorname{polylog}(n)$ bits.
        
      We are equipped with oracle access to a symmetric function $\sigma : \mathcal{U}
      \times \mathcal{U} \to \lbrace 0,1 \rbrace$ which determines
      the presence or absence of an edge between a particular pair
      of vertices based on their identifiers.
      So the edge set of the streamed graph is $E_\sigma = \lbrace \lbrace i, j \rbrace : \sigma(I_i, I_j) = 1 \rbrace$.
    \end{itemize}
\end{definition}

Implicit vertex streams can be understood as being defined over a
large pre-defined graph on $O(2^{\operatorname{polylog}(n)})$ vertices, so
that the stream describes the induced sub-graph arising by selecting $n$ of these vertices.

We also mention several special classes of graphs which have natural implicit representations or are easier to approximate MIS on.

\begin{definition}[Bounded Independence Graphs]\leavevmode
    \begin{itemize}
        \item A graph $G$ has \textbf{$f$-bounded independence} for some bounding function $f : \mathbb{N} \to \mathbb{N}$ if every $r$-neighborhood (the set of vertices within distance $r$) of every vertex of $G$ contains no independent set larger than $f(r)$.
        \item A family of graphs $\mathcal{F}$ has \textbf{polynomially bounded independence} if there is a polynomial $P$ such that every graph $G \in \mathcal{F}$ has $P$-bounded independence.
    \end{itemize}
\end{definition}

Bounded-independence graphs admit greedy approximation algorithms.
Consider the greedy algorithm that maintains an independent set
by arbitrarily ordering the vertices, and adding each subsequent
vertex so long as the set stays independent.
For an $f$-bounded independence graph this greedy procedure will give
an $f(1)$-approximation to MIS.
This algorithm can be implemented in any vertex stream (implicit or
explicit) by following the vertex arrival order and using $\tilde{O}(\alpha(G))$ space. 

\begin{definition}[Intersection Graphs]
  The \textbf{geometric intersection graph} $G=(V, E)$ of objects
      drawn from some geometric space $\lbrace X_1, X_2, \cdots, X_n
      \rbrace$ has one vertex associated with each object, and an edge
      between a pair of vertices if and only if the corresponding pair
      of objects intersect: $V = [n]$, $E = \lbrace \lbrace i, j
      \rbrace : X_i \cap X_j \neq \emptyset \rbrace$.
\end{definition}

For example, 
an $\mathbf{l^p_d}$\textbf{-ball graph} is the geometric intersection
graph of $n$ closed $l^p$ balls in $d$ dimensions, each of the form
$\mathcal{B}^p(x, r) = \lbrace y \in \mathbb{R}^d : \|x-y\|_p \leq r
\rbrace \subset \mathbb{R}^d$.
Of particular interest is the 
 \textbf{unit} $\mathbf{l^p_d}$\textbf{-ball graph}, where all the balls have radius $1$.
By specifying $p$ and $d$, we can define intersection graphs on disks,
spheres, squares, rectangles, hyper-cubes, and generalize to any combination of these and other geometric objects.
Many of these classes have polynomially bounded independence (as long as the ratio between the sizes of the smallest ball which can cover any of the objects and the largest ball which can fit inside any object is constant).
These geometric intersection graphs also form natural implicit vertex
streams, assuming that the geometry is discretised to allow $\operatorname{polylog}(n)$ sized representation of the shapes.

\subsection{Prior Work}\label{sec:prior-work}

\paragraph*{Edge Streams}
Halld\'{o}rsson {\em et al.}\ \cite{halldorsson2012streaming} showed that for general graphs in the
edge-arrival model $\Omega \left( \frac{n^2}{c^2\log^2{n}} \right)$
space is required to obtain a $c$-approximation to the maximum
independent set size (or maximum clique size) if $c = \Omega(\log n)$, and 
$\Omega \left( \frac{n^2}{c^4} \right)$ is required if $c =
o(\log{n})$ \cite{halldorsson2012streaming}. 
A corresponding $\tilde{O} \left( \frac{n^2}{c^2} \right)$ space random sampling algorithm shows this is tight up to logarithmic factors.
Braverman {\em et al.} \cite{blsvy18} showed that space $\Omega(\frac{m}{c^2})$ is needed, even if $c = o(\log n)$, where $m$ is the
number of edges of the input graph. This bound though only holds for small values of $m$.

\paragraph*{Explicit Vertex Streams}
The work of Halld\'{o}rsson {\em et al.}\
\cite{halldorsson2010streaming} gives an $O(n \log{n})$ space streaming algorithm which can find an independent set of expected size at least $\beta(G) = \sum_{v \in V} \frac{1}{\text{deg}(v) + 1}$. On general graphs, this only gives a $\Theta(n)$-approximation, but for polynomially bounded independence graphs, this gives a $\text{polylog}(n)$-approximation \cite{hk18}.

In our prior work, it was shown how to return
an estimate $\gamma \in \Omega\left( \frac{\beta(G)}{\log{n}} \right)$
with $\gamma \leq \alpha(G)$ from an explicit vertex arrival stream
using only $O(\log^3{n})$ space~\cite{Cormode:Dark:Konrad:18}. This result for example
gives a $O( \frac{\log^2 n}{\log \log n})$-approximation on unit interval graphs (see Figure~\ref{fig:explicit-table}).
However, the technique samples vertices based on their degree and so does not extend to implicit vertex streams.

Braverman {\em et al.} \cite{blsvy18} showed that in a vertex arrival model, where every vertex arrives together with all its
incident edges (as opposed to the explicit vertex stream model considered here where every vertex arrives together with its incident 
edges connecting to vertices that have previously arrived), space $\Omega(\frac{m}{c^3})$ is
required for computing a $c$-approximate MIS. In their construction the input graph has $\Theta(n c)$ edges, which thus
yields a lower bound of $\Omega(\frac{n}{c^2})$. Observe that our lower bound for explicit vertex streams is $\Omega(\frac{n^2}{c^7})$,
which is a quadratic improvement for constant $c$.

\paragraph*{Implicit Vertex Streams}
In \cite{emek2016space}, it was shown that it is possible to
$\frac{3}{2}$-approximate MIS for the intersection graph of a unit
interval stream using $\tilde{O}(\alpha(G))$ space.
In the same space, a $2$-approximation is possible for arbitrary
interval streams. Both these results are shown to be tight: any
$(\frac{3}{2}-\epsilon)$-approximation for unit intervals---or
$(2-\epsilon)$ for general intervals---requires $\Omega(n)$ space.
By clever use of sampling, the result can be adapted to provide
an approximation of $\alpha(G)$ of $\frac32 + \epsilon$ for unit
intervals and $2 + \epsilon$ for general intervals with space only
$\text{polylog}(n, \epsilon^{-1})$ \cite{cabello2017interval}.

\section{Bounds for Maximal Independent Set}

In this section, we consider streaming algorithms for the 
maximal independent set problem.
Vertex-arrival streams (both explicit and implicit) are well-suited to the maximal independent set problem, since they
allow the implementation of the \textsc{Greedy} algorithm for independent sets, which greedily 
add every incoming vertex $v$ to an initially empty independent set $I$ if this is possible, i.e., if $I \cup \{v \}$ 
is an independent set. This yields the following result:

\begin{fact}
 The \textsc{Greedy} algorithm for independent sets is a one-pass $O(n \log n)$ space maximal independent set algorithm 
 in the vertex-arrival order streaming model (for both implicit and explict vertex streams).
\end{fact}

Since the only space required by the algorithm is to store a valid independent set, the storage is in fact bounded by $\tilde{O}(\alpha(G))$, 
the space needed to store an MIS, which is bounded by $O(n \log n)$. 

This fact raises the question how well we can solve the maximal independent set problem in edge-arrival streams, which we address in the
remainder of this section. We will first show that computing a maximal independent set in one pass in the edge-arrival model is not possible
using sublinear space, i.e., space $\Omega(n^2)$ is required. We then ask whether we can compute an independent set that is 
{\em approximately maximal} in a single pass: 

\begin{definition}[Approximate Maximality]
 Let $G=(V, E)$ be an $n$-vertex graph, and let $I \subseteq V$ be an independent set. Then $I$ is $\delta$-maximal, if $|I \cup \Gamma_G[I]| \ge \delta n$. 
\end{definition}

A $\delta$-maximal independent set $I$ covers a $\delta$-fraction of the vertices, or, in other words, when removing $I$ and
its neighbors $\Gamma_G[I]$ from the graph, then $(1-\delta)n$ vertices are remaining. We will next show that establishing
approximate maximality in edge-arrival streams requires strictly more space that computing a maximal independent set in vertex-arrival 
streams (i.e., $\omega(n)$ space), even if $\delta = \frac{24}{25}$. Regarding stronger approximate maximality, 
our lower bound yields that computing a $(1-\frac{1}{n^\epsilon})$-maximal independent set requires space 
$\Omega(n^{2-o(1)})$, for every $\epsilon > 0$.

Interestingly, if we allow an algorithm to perform multiple passes, then sublinear space algorithms for maximal independent set 
can be obtained. Such algorithms are in fact immediately implied by the correlation clustering algorithms given
in Ahn et al. \cite{acgmw15}. Their result yields the following theorem:

\begin{theorem}
 There is a $O(\log \log n)$-pass streaming algorithm for maximal independent set that uses space $\tilde{O}(n)$.
\end{theorem}

\subsection{Lower Bound for Maximal Independent Set in Edge-arrival Streams}
We give a reduction to the well-known two-party communication complexity problem \textsc{Index}:
\begin{definition}
 In the two-party communication problem \textsc{Index}, Alice holds an $N$-bit string $X \in \{0, 1\}^N$ and Bob holds
 an index $\sigma \in [N]$. Alice sends a single message to Bob, who, upon receipt, outputs $X_\sigma$.
\end{definition}
 
It is well-known that Alice essentially needs to send all $N$ bits to Bob:

\begin{theorem}[\cite{knr99}] \label{thm:index}
 The randomized constant error communication complexity of \textsc{Index} is $\Omega(N)$.
\end{theorem}

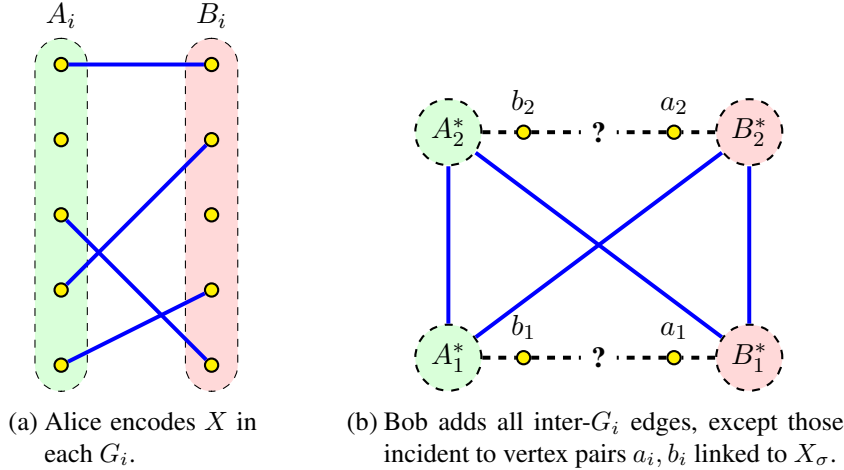
\begin{figure}[h]
    \centering
    \subcaptionbox{Alice encodes $X$ in each $G_i$.}[0.2\textwidth]{
        \begin{tikzpicture}[thick,
          vertex/.style={draw, circle, thick, fill=yellow, inner sep=0pt, minimum size=5pt},
          edge/.style={blue, line width=1.5pt}
        ]
            \foreach \y in {1,...,5}
            {
                \node[vertex] (A-\y) at (1, \y) {};
                \node[vertex] (B-\y) at (3, \y) {};
            }
            \draw[edge] (A-1) to (B-2);
            \draw[edge] (A-2) to (B-4);
            \draw[edge] (A-3) to (B-1);
            \draw[edge] (A-5) to (B-5);
            \begin{pgfonlayer}{bg}
                \node[draw, fill=green!15!white, rounded corners=10pt, dashed, fit=(A-1) (A-5), inner sep=7pt, label=above:$A_i$] {};
                \node[draw, fill=red!15!white, rounded corners=10pt, dashed, fit=(B-1) (B-5), inner sep=7pt, label=above:$B_i$] {};
            \end{pgfonlayer}
        \end{tikzpicture}
    }
    \hspace{1cm}
    \subcaptionbox{Bob adds all inter-$G_i$ edges, except those incident to vertex pairs $a_i, b_i$ linked to $X_\sigma$.}[0.4\textwidth]{
        \begin{tikzpicture}[thick,
          vertex/.style={draw, circle, thick, fill=yellow, inner sep=0pt, minimum size=5pt},
          A/.style={draw, circle, dashed, fill=green!15!white, inner sep=0pt, minimum size=25pt},
          B/.style={draw, circle, dashed, fill=red!15!white, inner sep=0pt, minimum size=25pt},
          edge/.style={dashed, line width=1.5pt},
          superedge/.style={blue, line width=1.5pt}
        ]
            \foreach \y in {1,2}
            {
                \node[vertex, label=above:$b_\y$] (b-\y) at (2, 3*\y) {};
                \node[A] (A-\y) at (1, 3*\y) {$A^\ast_\y$};
                \draw[edge] (A-\y) to (b-\y);
                \node[B] (B-\y) at (5, 3*\y) {$B^\ast_\y$};
                \node[vertex, label=above:$a_\y$] (a-\y) at (4, 3*\y) {};
                \node (q-\y) at (3, 3*\y) {\textbf{?}};
                \draw[edge] (B-\y) to (a-\y);
                \draw[edge] (a-\y) to (q-\y);
                \draw[edge] (b-\y) to (q-\y);
            }
            \draw[superedge] (A-1) to (A-2);
            \draw[superedge] (A-1) to (B-2);
            \draw[superedge] (B-1) to (B-2);
            \draw[superedge] (B-1) to (A-2);
        \end{tikzpicture}
    }
    \caption{Construction of the gadget used to prove theorem~\ref{thm:max-is-lb}.}
    \label{fig:maximal-gadget}
\end{figure}

\begin{theorem} \label{thm:max-is-lb}
 Every randomized constant error one-pass streaming algorithm in the edge arrival model that computes a maximal independent set requires $\Omega(n^2)$ space.
\end{theorem}
\begin{proof}
 For an integer $n$, consider an instance $(X, \sigma)$ of \textsc{Index} with $N = n^2$. Given $X$, Alice constructs
 the input graph $G= G_1 \stackrel{.}{\cup} G_2$, which is the disjoint union of graphs $G_1$ and $G_2$, where $G_2$ 
 is a disjoint copy of $G_1$. Graph $G_1(A_1, B_1, E_1)$ is a bipartite graph with $|A_1| = |B_1| = n$. 
 The vector $X \in \{0, 1\}^{n^2}$ is seen as the $n \times n$ incident matrix of the bipartite graph $G_1$, which defines
 the edge sets of $G_1$ (and thus also of $G_2$). Alice simulates a streaming algorithm for maximal independent set
 on the stream of edges of $G_1$ and $G_2$ and sends the memory state to Bob.
 
 Given index $\sigma$, Bob identifies the pair of vertices $(a_1, b_1) \in  A_1 \times B_1$ that corresponds to the index $\sigma$. Let $(a_2, b_2)$ be the copy of $(a_1, b_1)$ in $G_2$. Let $A^\ast_1 = A_1 \setminus \{a_1\}$, similarly for $A^\ast_2, B^\ast_1, B^\ast_2$. Given the message from Alice, i.e., the memory state of the run of the maximal independent set streaming algorithm, Bob continues the algorithm by adding all edges between the two sets of vertices $A^\ast_1 \cup B^\ast_1$ and $A^\ast_2 \cup B^\ast_2$. This cannot cause a conflict, because no edges were previously present between $G_1$ and $G_2$. The construction is shown in figure~\ref{fig:maximal-gadget}.
 
 Let $I$ be the output maximal independent set computed by the algorithm. We argue that Bob can decide whether edge $(a_1, b_1)$ is in the input graph and hence determine the value of $X_\sigma$. First, suppose that $X_\sigma = 1$. Then, for every $j  \in \{1, 2\}$, $\{a_j, b_j \} \not\subseteq I$, i.e., $a_j$ and $b_j$ are not both included in $I$. We argue now that if $X_\sigma = 0$, then $\{a_1, b_1 \} \subseteq I$ or $\{a_2, b_2 \} \subseteq I$. Suppose that this is not the case, i.e., there are vertices $x_1 \in \{a_1, b_1 \}$ and $x_2  \in \{a_2, b_2 \}$ which are not included in $I$. Since $I$ is maximal, there exists a vertex $y_1 \in I \cap \Gamma_G(x_1)$
 and $y_2 \in I \cap \Gamma_G(x_2)$. Furthermore, since $a_1, b_1, a_2, b_2$ can only have edges to $B^\ast_1, A^\ast_1, B^\ast_2, A^\ast_2$ respectively, we have 
 $y_1 \in A^\ast_1 \cup B^\ast_1$ and $y_2 \in A^\ast_2 \cup B^\ast_2$. However, since Bob added all edges between these vertex sets, there must be an edge between 
 $y_1$ and $y_2$, which contradicts $I$ being an IS. Hence, if $X_\sigma = 0$, 
 then $\{a_1, b_1 \} \in I$ or $\{a_2, b_2 \} \in I$ and we can thus determine the value of $X_\sigma$.
 The space required by the streaming algorithm is therefore at least the communication complexity of \textsc{Index},
 which is $\Omega(N) = \Omega(n^2)$, by Theorem~\ref{thm:index}. 
\end{proof}

\begin{figure}[h]
    \centering
    \hspace{1cm}
    \subcaptionbox{When $X_\sigma = 1$, clearly $a_i$ and $b_i$ cannot both be in any IS. So neither pair $\lbrace a_i, b_i \rbrace$ is in output.}[0.3\textwidth]{
        \begin{tikzpicture}[thick,
          vertex/.style={draw, circle, thick, fill=yellow, inner sep=0pt, minimum size=5pt},
          A/.style={draw, circle, dashed, fill=green!15!white, inner sep=0pt, minimum size=25pt},
          B/.style={draw, circle, dashed, fill=red!15!white, inner sep=0pt, minimum size=25pt},
          edge/.style={dashed, line width=1.5pt},
          superedge/.style={blue, line width=1.5pt},
          boldedge/.style={black, line width=1.5pt}
        ]
            \foreach \y in {1,2}
            {
                \node[vertex, label=above:$b_\y$] (b-\y) at (2, 3*\y) {};
                \node[A] (A-\y) at (1, 3*\y) {$A^\ast_\y$};
                \draw[edge] (A-\y) to (b-\y);
                \node[B] (B-\y) at (4, 3*\y) {$B^\ast_\y$};
                \node[vertex, label=above:$a_\y$] (a-\y) at (3, 3*\y) {};
                \draw[edge] (B-\y) to (a-\y);
                \draw[superedge] (a-\y) to (b-\y);
            }
            \draw[boldedge] (A-1) to (A-2);
            \draw[boldedge] (A-1) to (B-2);
            \draw[boldedge] (B-1) to (B-2);
            \draw[boldedge] (B-1) to (A-2);
        \end{tikzpicture}
    }
    \hspace{1cm}
    \subcaptionbox{When $X_\sigma = 0$, in any IS only one pair can be covered by neighbours. So any maximal IS output contains a pair $\lbrace a_i, b_i \rbrace$.}[0.3\textwidth]{
        \begin{tikzpicture}[thick,
          vertex/.style={draw, circle, thick, fill=yellow, inner sep=0pt, minimum size=5pt},
          A/.style={draw, circle, dashed, fill=green!15!white, inner sep=0pt, minimum size=25pt},
          B/.style={draw, circle, dashed, fill=red!15!white, inner sep=0pt, minimum size=25pt},
          edge/.style={dashed, line width=1.5pt},
          superedge/.style={blue, line width=1.5pt},
          boldedge/.style={black, line width=1.5pt}
        ]
            \foreach \y in {1,2}
            {
                \node[vertex, label=above:$b_\y$] (b-\y) at (2, 3*\y) {};
                \node[A] (A-\y) at (1, 3*\y) {$A^\ast_\y$};
                \draw[edge] (A-\y) to (b-\y);
                \node[B] (B-\y) at (4, 3*\y) {$B^\ast_\y$};
                \node[vertex, label=above:$a_\y$] (a-\y) at (3, 3*\y) {};
                \draw[edge] (B-\y) to (a-\y);
            }
            \draw[boldedge] (A-1) to (A-2);
            \draw[boldedge] (A-1) to (B-2);
            \draw[boldedge] (B-1) to (B-2);
            \draw[boldedge] (B-1) to (A-2);
        \end{tikzpicture}
    }
    \caption{Sketch proof for theorem~\ref{thm:max-is-lb} that at least one pair $a_i, b_i$ is in any maximal IS when $\mathbf{X}_\sigma = 0$.}
    \label{fig:maximal-proof}
\end{figure}
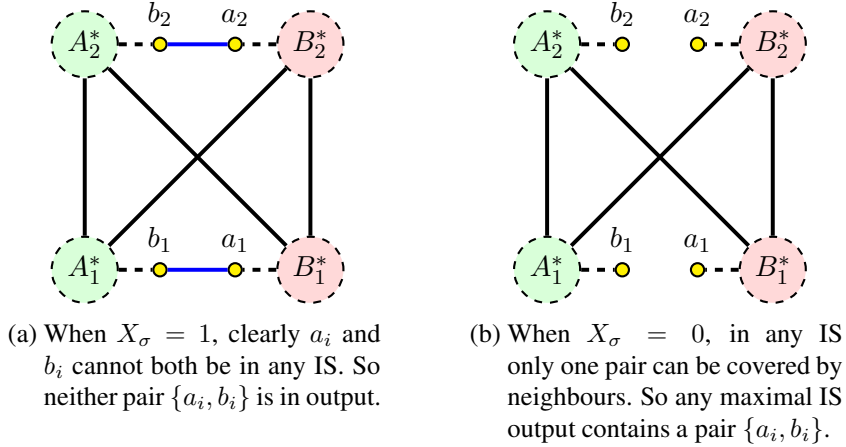

\subsection{Lower Bound for Approximate Maximality in Edge-arrival Streams}
We now extend the lower bound given in the previous subsection to approximate maximality. Central to our construction
are {\em Ruzsa-Szemer\'{e}di graphs}, which have previously been used for the construction of space lower bounds for streaming
algorithms for the maximum matching problem \cite{gkk12,k15,akly16}:

\begin{definition}[Ruzsa-Szemer\'{e}di graph]
 A bipartite graph $G$ is an {\em $(r,s)$-Ruzsa-Szemer\'{e}di graph} if its edge set can be partitioned into
 $r$ induced matchings each of size $s$. 
\end{definition}
Recall that a matching $M \subseteq E$ in a graph $G=(V, E)$ is induced, if the edge set of the vertex-induced
subgraph $G[V(M)]$ equals $M$, i.e., there are no other edges interconnecting $V(M)$ different from $M$.

Our lower bound for approximate maximality is obtained by a reduction to the two-party communication problem
\textsc{RS-Index}, defined as follows:

\begin{definition}[\textsc{RS-Index}]
 Let $H$ be an $(r,s)$-Ruzsa-Szemer\'{e}di graph with induced matchings $M_1, M_2, \dots, M_r$. For each induced 
 matching $M_i$, let $M_i' \subseteq M_i$ be a uniform random subset of size $s/2$ (we assume that $s$ is even). 
 The \textsc{RS-Index} problem is a one-way two-party communication problem, where $H$, and, in particular, $M_1, M_2, \dots, M_r$ are 
 known by both parties. In addition, 
 Alice holds the graph $G = H[ \cup_i M_i']$, and Bob holds a uniform random index $i \in \{1, 2, \dots, r \}$. Alice sends a single message 
 to Bob, who, upon receipt, outputs at least $C \cdot s$ edges of $M_i'$, for an arbitrary small constant $C$.
\end{definition}

Observe that this problem is similar in spirit to \textsc{Index}: In \textsc{Index}, Bob needs to learn one uniform random bit, while
in \textsc{RS-Index}, Bob needs to learn the presence of many edges of $M_i'$. A lower bound on the communication complexity
of \textsc{RS-Index} is implicit in \cite{gkk12} \footnote{In \cite{gkk12} a lower bound is given for the task of computing
a maximum matching. Their hardness stems from the fact that it is hard to learn many edges of $M_i'$ under the distribution described
in the definition of \textsc{RS-Index}.}:

\begin{theorem}[\cite{gkk12}] \label{thm:rs-index}
 The randomized constant error communication complexity of \textsc{RS-Index} is $\Omega(r \cdot s)$.
\end{theorem}

Equipped with the \textsc{RS-Index} problem, we now give a reduction from approximate maximality to \textsc{RS-Index},
which yields our lower bound for streaming algorithms:

\begin{lemma}\label{lem:approx-lb}
 Let $r,s,n$ be integers such that there is an $n$-vertex $(r,s)$-Ruzsa-Szemer\'{e}di graph. Then, 
 every randomized constant error one-pass streaming algorithm in the edge arrival model that computes a 
 $(1- \frac{s}{6n})$-maximal independent set requires $\Omega(r \cdot s)$ space.
\end{lemma}
\begin{proof}
 Let $H$ be an $n$-vertex $(r,s)$-Ruzsa-Szemer\'{e}di graph, and let $G$ be Alice's input graph for the \textsc{RS-Index} problem
 derived from $H$. Let $M_1, M_2, \dots, M_r$ denote the induced matchings in $H$, let $V_i = V(M_i)$, and let $M_i' \subseteq M_i$
 denote the subset of edges of matching $M_i$ that is included in $G$. 
 Let $i$ be Bob's input. Furthermore, let $\mathcal{A}$ be a constant error randomized one-pass streaming algorithm for the edge-arrival 
 model that computes a $(1-\frac{s}{6N})$-maximal independent set on a graph on $N$ vertices. We now show how $\mathcal{A}$ can be 
 used to solve $\textsc{RS-Index}$: 
 
 Given $G$, let $\tilde{G}$ be the graph obtained from $G$, where every induced matching $M_i'$ in $G$ is replaced by 
 edges $\tilde{M_i'} :=  M_i \setminus M_i'$ (observe that $E(G) \cup E(\tilde{G}) = E(H)$). Alice now constructs two disjoint 
 copies $G_1$ and $G_2$ of $\tilde{G}$, runs algorithm $\mathcal{A}$ on $G_1 \stackrel{.}{\cup} G_2$ (on an arbitrary ordering of their edges), and sends 
 the memory state to Bob. Bob constructs the edge set $F$ that connects every vertex $v_1 \in V(G_1) \setminus V_{i1}$ with every vertex
 $v_2 \in V(G_2) \setminus V_{i2}$,
 where $V_{i1}$ and $V_{i2}$ are the copies of the vertices $V_i$ in graphs $G_1$ and $G_2$, respectively, 
 and continues the execution of $\mathcal{A}$ on $F$. Let $I$ be the independent set produced by algorithm $\mathcal{A}$.
 
 Observe that the graph processed by algorithm $\mathcal{A}$ contains $N = 2n$ vertices. Since $I$ is $(1- \frac{s}{6N})$-maximal, 
 we have $|V \setminus \Gamma[I]| \le N - (1- \frac{s}{6N})N = s/6$. This allows us to identify $\Omega(s)$ edges of $M'_i$ as follows: 
 
 Let $a,b$ be the incident vertices to an arbitrary edge of $M_i'$, let $a_1,b_1$ be the copies of $a,b$ in $G_1$, and let $a_2, b_2$ be the 
 copies of $a,b$ in $G_2$. Observe that $a_1$ and $b_1$ are not connected in $G_1$, and $a_2$ and $b_2$ are not connected in $G_2$.
 We now claim that if all vertices $a_1, b_1, a_2, b_2$ are covered by $I$, i.e., $\{a_1, b_1, a_2, b_2 \} \subseteq \Gamma[I]$, then either 
 $\{a_1, b_1 \} \subseteq I$ or $\{a_2, b_2 \} \subseteq I$ (or both).
 Indeed, suppose that this is not the case. Then there are vertices $x_1 \in \{a_1, b_1 \}$ and $x_2 \in \{a_2, b_2 \}$ with $x_1, x_2 \notin I$. 
 Let $y_1 \in I$ be a vertex incident to $x_1$, and let $y_2 \in I$ be a vertex incident to $x_2$. By the construction of the input graph,
 $y_1 \in V(G_1) \setminus V_{i1}$, and $y_2 \in V(G_2) \setminus V_{i2}$. Observe however that the edge $y_1y_2$ was included by Bob, which implies
 that $y_1, y_2$ are not independent, a contradiction. Hence, either $\{a_1, b_1 \} \subseteq I$ or $\{a_2, b_2 \} \subseteq I$ (or both) holds. 
 Observe that this implies that the algorithm identified that there is no edge between $a_1, b_1$, which in turn implies that we learned one edge of $M_i'$.
 Hence, for every pair of vertices $a,b$ of $M_i'$, either at least one vertex among $\{a_1, b_1, a_2, b_2 \}$ is not covered by $I$, or we learn one edge of $M_i'$. Since 
 there are $s/2$ edges in $M_i'$, and at most $s/6$ vertices of the input graph are not covered by $I$, we learn at least $s/2 - s/6 = \Omega(s)$ edges 
 of $M_i'$, which thus solves \textsc{RS-Index}. By Theorem~\ref{thm:rs-index}, algorithm $\mathcal{A}$ therefore requires space $\Omega(r \cdot s)$.

\end{proof}

In \cite{gkk12} it is shown that there are $n$-vertex $(n^{\Theta(\frac{1}{\log \log n})}, (\frac{1}{4} - \epsilon)n)$ Ruzsa-Szemer\'{e}di 
graphs, for every $\epsilon > 0$, and in \cite{ams12}, it is shown that there are $n$-vertex
Ruzsa-Szemer\'{e}di graphs with $\Theta(n^{2-o(1)})$ edges such that each matching is of size $n^{1-o(1)}$. Combined with 
Lemma~\ref{lem:approx-lb}, we obtain our main theorem: 

\begin{theorem}
 Every randomized constant error one-pass streaming algorithm that computes a $\frac{24}{25}$-maximal independent set requires space
 $n^{1+ \Omega(\frac{1}{\log \log n})}$, and every such algorithm computing a $(1-\frac{1}{n^{\epsilon}})$-maximal independent 
 set requires space $\Omega(n^{2 - o(1)})$, for every $\epsilon > 0$.
\end{theorem}

\section{Maximum Independent Set in Explicit Vertex Streams}\label{sec:general-explicit}

We first introduce and show the hardness of a ``chained index''
problem, which we then use to show the hardness of approximating the
size $\alpha(G)$ (and hence also for finding an MIS). 

\subsection{Chained Index Communication Problem}

We define a multi-party communication problem, which allows us to
prove new lower bounds on several streaming problems.
The problem is closely related to pointer jumping and generalizes
the classic $2$-party index communication problem
to more parties by ``chaining'' together multiple
instances which have the same answer but are otherwise
independent.
In our setting, each party (other than the last) holds a vector
that contains (somewhere) a bit which is ``the answer'' to the instance.
Each party (other than the first) knows where the answer bit is
located in the previous party's vector.
Communication is one-way and private, with each player receiving a message from the previous player and then sending a message to the next player.
This rules out any trivial solution where a party can simply look up
the bit announced by a later party. 
Formally, 

\begin{definition}
    The $k$-party chained index problem $\textsc{Chain}_k$ consists of
    $(k-1)$ $n$-bit binary vectors $\lbrace X^{(i)} \rbrace_{i =
      1}^{k-1}$, along with corresponding indices $\lbrace \sigma_i
    \rbrace_{i = 1}^{k-1}$ from the range $[n]$.
    We have the promise that
    the entries $\lbrace X^{(i)}_{\sigma_i} \rbrace_{i = 1}^{k-1}$ are
    all equal to the desired answer bit $z \in \lbrace 0, 1 \rbrace$.
    The input is initially allocated as follows:
    \begin{itemize}
        \item The first party $P_1$ knows $X^{(1)}$
        \item Each intermediate party $P_p$ for $1 < p < k$ knows $X^{(p)}$ and $\sigma_{p-1}$
        \item The final party $P_k$ knows just $\sigma_{k-1}$
    \end{itemize}
    Communication proceeds as follows: $P_1$ sends a single message to
    $P_2$, then $P_2$ communicates to $P_3$, and so on, with each
    party sending exactly one message to its immediate
    successor. After all messages are sent, $P_k$ must correctly
    output $z$, succeeding with probability at least $2/3$.
    If the promise condition is violated, any output is considered correct.
\end{definition}

There is a trivial communication upper bound of $O(n)$ bits: for
instance, simply have the penultimate party send $X^{(k-1)}$ to the
final party who can then return $X^{(k-1)}_{\sigma_{k-1}}$.
We now show lower bounds by a reduction from a different multi-party
communication problem.

\begin{definition}[\cite{damm1998some}]
    The Boolean conservative one-way $k$-party pointer jumping problem
    $\textsc{Jump}_k$ consists of a constant $\alpha \in [n]$ and
    $k-1$ functions $\lbrace f_i \rbrace_{i = 2}^k$.
    The first $k-2$ are of the form $f_i : [n] \to [n]$, and the final
    one is of the form $f_{k-1} : [n] \to \lbrace 0,1 \rbrace$.
    We use $f_{i:j}$ to refer to the composition\footnote{Using the convention that $(g \circ h)(x) = h(g(x))$.} of functions $f_i \circ f_{i+1} \circ \cdots \circ f_{j-1} \circ f_j$. The input is divided as follows:
    \begin{itemize}
        \item The first party $P_1$ knows all the functions $\lbrace f_i \rbrace_{i = 2}^k$
        \item The second party $P_2$ knows $\alpha$ and every $f_j$ for $j \geq 3$
        \item Each other party $P_i$ knows $f_{2:i-1}(\alpha)$ and every $f_j$ for $j \geq i + 1$
    \end{itemize}
    
Each party sends exactly one message in ascending order to their
immediate successor, i.e. $P_1$ sends to $P_2$, then $P_2$ sends to $P_3$, and so on. After all messages are sent, $P_k$ must correctly output $f_{2:k}(\alpha)$ with probability at least $2/3$.
\end{definition}

The conservative version of one-way $k$-party pointer jumping problem
was introduced and studied in \cite{damm1998some}, showing
$\Omega(\frac{n}{k^2})$ hardness for $k \in o(n^{\frac{1}{3}})$ for a
version with a non-Boolean final layer.
Later, \cite{chakrabarti2007lower} extended this to all $k$ and to the Boolean version.

\begin{theorem}[Theorem~2 in \cite{chakrabarti2007lower}]\label{thm:pointer-jumping}
    Any communication scheme  $\mathcal{A}$  which solves
    $\textsc{Jump}_k$ must communicate at least $\Omega(\frac{n}{k^2})$ bits.
\end{theorem}

Hardness for our new problem is then as follows.

\begin{theorem}
  Any communication scheme $\mathcal{B}$ which solves
  $\textsc{Chain}_k$ 
  must communicate at least $\Omega(\frac{n}{k^2})$ bits.
\end{theorem}
\begin{proof}
    We prove the claim by showing that any instance of
    $\textsc{Jump}_k$ can be reduced to an instance of
    $\textsc{Chain}_k$ without any communication.
    Hence, any algorithm which solves $\textsc{Chain}_k$ can solve
    $\textsc{Jump}_k$ with no change in the communication cost.
    Combining this with the lower bound of Theorem~\ref{thm:pointer-jumping} gives the result.
    
    Fix an instance of $\textsc{Jump}_k$.
    For each $i$ let $X^{(i)}$ be the binary vector whose
    $j$\textsuperscript{th} entry is $f_{i+1:k}(j)$.
    For each $i$ let $\sigma_i = f_{2:i}(\alpha)$.
    Now we observe three facts:
    \begin{itemize}
        \item Every $\lbrace X^{(i)}_{\sigma_i} \rbrace_{i=1}^{k-1}$ is equal to $f_{2:k}(\alpha)$
        \item Each party $P_i$ for $i < k$ knows all the information required to compute $X^{(i)}$
        \item Each party $P_i$ for $i > 1$ knows all the information required to compute $\sigma_{i-1}$
    \end{itemize}
    So we have constructed (with no communication) a $k$-party chained index problem which, if solved, will tell us exactly $f_{2:k}(\alpha)$.
    It therefore follows that the commnication cost for any solution to
   $\textsc{Chain}_k$ is at least that for $\textsc{Jump}_k$.
\end{proof}

In particular, for constant $k$, we have a tight bound on the communication complexity of the $k$-party chained index problem of $\Theta(n)$.
We conjecture that a dependence on $k$ is not necessary. 

\begin{conjecture}
    Any communication scheme which solves $\textsc{Chain}_k$ requires $\Omega(n)$ communication.
\end{conjecture}

If proven, this would give a hardness of $\Omega(\frac{n^2}{c^5})$ for theorem~\ref{thm:explicit-vertex}.

Although the two problems of $\textsc{Jump}_k$ and $\textsc{Chain}_k$
may look fairly similar, we find that the form of $\textsc{Chain}_k$ is
much more convenient to show lower bounds for independent set and
related problems, as we show in the subsequent sections.

\subsection{MIS hardness in explicit vertex streams}

We show a new lower bound for the vertex streaming space complexity of approximate maximum independent set.

\begin{theorem}\label{thm:explicit-vertex}
Any algorithm for the explicit vertex stream model which finds a $c$-approximation to $\alpha(G)$ with probability at least $2/3$ requires $\Omega\left( \frac{n^2}{c^7} \right)$ space.
\end{theorem}

For ease of argument, we will actually prove an equivalent result for
the problem of clique number approximation, and then note that the
complement of the constructed graph can be used with the same
arguments to prove Theorem~\ref{thm:explicit-vertex}.
To see this equivalence, note that an MIS of a graph is a maximum
clique in its complement, and in a vertex streaming model an algorithm
can simulate operation on the complement by taking the complement of
each vertex as it arrives.
Importantly, 
in the edge-arrival model this reduction is not possible, since the
model does not allow  to `subtract' the observed edges from the complete graph.

\begin{theorem}\label{thm:explicit-vertex-anti}
Any algorithm for the explicit vertex stream model which finds a $c$-approximation to the size of the largest clique $\omega(G)$ with probability at least $2/3$ requires $\Omega\left( \frac{n^2}{c^7} \right)$ space.
\end{theorem}

The heart of our construction is to use an erasure code to encode a length
$O(\frac{n^2}{c^4})$ binary vector on $O(\frac{n}{c})$ vertices, with
each bit corresponding to the presence or absences of a clique of size
$2c$.
The use of the erasure code is to ensure that no pair of these cliques
can share an edge. We can then chain together $2c$ such gadgets to
encode an instance of $\textsc{Chain}_{2c}$ such that if the correct
answer is $1$, the resulting graph has an independent set of size
$4c^2$, while if the correct answer is $0$ the graph has no
independent sets of size larger than $4c - 1$. Any (one-sided)
$c$-approximation algorithm could distinguish these two cases, and so
the result is proved.

First we define our clique gadget.

\begin{lemma}\label{lem:erasure}
    For any positive integers $n$ and $c^2 < \frac{n}{8}$, there exists a graph on $n$ vertices containing $\frac{n^2}{16c^2}$ edge-disjoint cliques of size $2c$ and no cliques of size larger than $2c$.
\end{lemma}

\begin{figure}[h]
    \begin{center}
        \begin{subfigure}{0.3\textwidth}
            \resizebox{\textwidth}{!}{
                \begin{tikzpicture}[thick,
                  vertex/.style={draw, circle, thick, fill=yellow, inner sep=0pt, minimum size=5pt},
                  edge/.style={draw=red, line width=2pt}
                ]
                    \foreach \x in {1,...,5}
                    \foreach \y in {1,2,3,5}
                    {
                        \node[vertex] (\x-\y) at (\x, \y) {};
                    }
                    \foreach \x in {1,3,5}
                    {
                        \node (\x-4) at (\x, 4.1) {\Huge $\vdots$};
                    }
                    \draw[decoration={brace, raise=15pt}, decorate] (1-5) to node[above=16pt] {$p$} (5-5);
                    \draw[edge] (1-1) to (2-2);
                    \begin{pgfonlayer}{bg}
                        \foreach \y in {1,...,3}
                        {
                            \node[draw, rounded rectangle, dashed, fit=(1-\y) (5-\y), inner sep=7pt, label=right:$V_\y$] {};
                        }
                        \node[draw, rounded rectangle, dashed, fit=(1-5) (5-5), inner sep=7pt, label=right:$V_{2c}$] {};
                    \end{pgfonlayer}
                \end{tikzpicture}
            }
            \caption{Select an edge.}
        \end{subfigure}
        \begin{subfigure}{0.3\textwidth}
            \resizebox{\textwidth}{!}{
                \begin{tikzpicture}[thick,
                  vertex/.style={draw, circle, thick, fill=yellow, inner sep=0pt, minimum size=5pt},
                  edge/.style={draw=red, line width=2pt}
                ]
                    \foreach \x in {1,...,5}
                    \foreach \y in {1,2,3,5}
                    {
                        \node[vertex] (\x-\y) at (\x, \y) {};
                    }
                    \foreach \x in {1,3,5}
                    {
                        \node (\x-4) at (\x, 4.1) {\Huge $\vdots$};
                    }
                    \draw[decoration={brace, raise=15pt}, decorate] (1-5) to node[above=16pt] {$p$} (5-5);
                    \draw[edge] (1-1) to (2-2);
                    \draw[edge] (2-2) to (3-3);
                    \draw[edge] (3-3) to (5-5);
                    \begin{pgfonlayer}{bg}
                        \foreach \y in {1,...,3}
                        {
                            \node[draw, rounded rectangle, dashed, fit=(1-\y) (5-\y), inner sep=7pt, label=right:$V_\y$] {};
                        }
                        \node[draw, rounded rectangle, dashed, fit=(1-5) (5-5), inner sep=7pt, label=right:$V_{2c}$] {};
                    \end{pgfonlayer}
                \end{tikzpicture}
            }
            \caption{Extend linearly to size $2c$.}
        \end{subfigure}
        \begin{subfigure}{0.3\textwidth}
            \resizebox{\textwidth}{!}{
                \begin{tikzpicture}[thick,
                  vertex/.style={draw, circle, thick, fill=yellow, inner sep=0pt, minimum size=5pt},
                  edge/.style={draw=gray, line width=1pt}
                ]
                    \foreach \x in {1,...,5}
                    \foreach \y in {1,2,3,5}
                    {
                        \node[vertex] (\x-\y) at (\x, \y) {};
                    }
                    \foreach \x in {1,3,5}
                    {
                        \node (\x-4) at (\x, 4.1) {\Huge $\vdots$};
                    }
                    \draw[decoration={brace, raise=15pt}, decorate] (1-5) to node[above=16pt] {$p$} (5-5);
                    \foreach \a in {1,...,5}
                    \foreach \b in {1,...,5}
                    {
                        \draw[edge] (\a-1) to (\b-2);
                    }
                    \begin{pgfonlayer}{bg}
                        \foreach \y in {1,...,3}
                        {
                            \node[draw, rounded rectangle, dashed, fit=(1-\y) (5-\y), inner sep=7pt, label=right:$V_\y$] {};
                        }
                        \node[draw, rounded rectangle, dashed, fit=(1-5) (5-5), inner sep=7pt, label=right:$V_{2c}$] {};
                    \end{pgfonlayer}
                \end{tikzpicture}
            }
            \caption{Can make $p^2$ such cliques.}
        \end{subfigure}
        \caption{Clique gadget construction in lemma~\ref{lem:erasure}.}\label{fig:clique-gadget-construction}
    \end{center}
\end{figure}
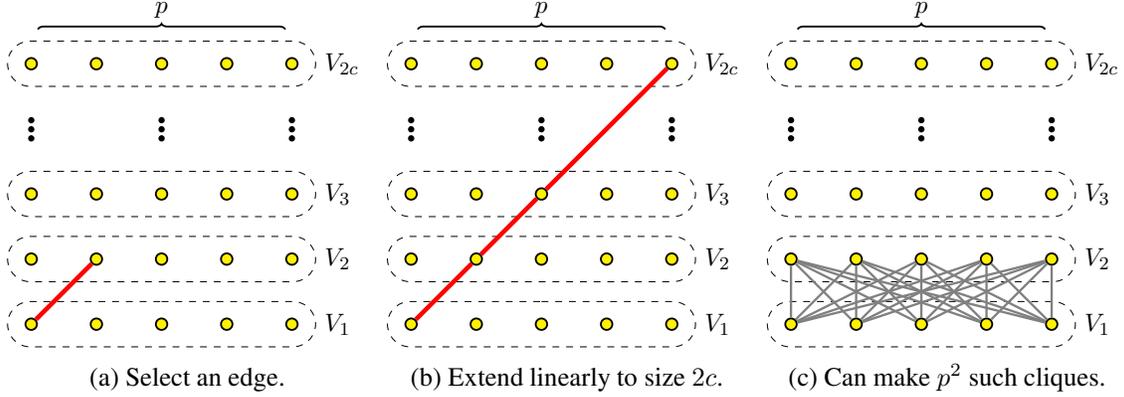

\begin{proof}
    We construct the sets from an erasure code with block size $2c$ and message size $2$. Choose a prime $p$ such that $\frac{n}{4c} \leq p \leq \frac{n}{2c}$ (which is guaranteed to exist). Now take $2c < p$ groups of vertices, each of size $p$. Label the groups $V_i$ (for $i \in [2c]$) and label the items in each group $V_i$ as $v^i_j$ (for $j \in [p]$). Leftover vertices are added to the final graph as isolated vertices.
    
    For each polynomial $\mathcal{P} \in \text{GF}(p^2)$ we define $K_\mathcal{P}$ to be the clique over vertices $\lbrace v^i_{\mathcal{P}(i)} | i \in [2c] \rbrace$. This can be viewed as taking each of the $p^2$ possible edges between $V_1$ and $V_2$ and extending them ``linearly'' to the other layers (see Figure~\ref{fig:clique-gadget-construction}).
    Clearly $\mathcal{K} = \lbrace K_\mathcal{P} | \mathcal{P} \in
    \text{GF}(p^2) \rbrace$ consists of $p^2 > \frac{n^2}{16c^2}$
    cliques, each of size $2c$.
We next show that they are pairwise edge-disjoint and that their union contains no larger cliques.
    
    Each clique contains exactly one vertex from each group $V_i$, so
    for two cliques to share an edge there must be distinct
    polynomials $\mathcal{P}, \mathcal{Q} \in GF(p^2)$ that have the
    same value at two different points: $\mathcal{P}(i) =
    \mathcal{Q}(i)$ and $\mathcal{P}(j) = \mathcal{Q}(j)$ for $i \neq
    j$  --- a contradiction.  
    Finally, because no clique contains a pair of vertices from a single $V_i$, their union can contain no internal edges on any $V_i$. So any clique can contain at most $1$ vertex from each $V_i$, giving a maximum size of $2c$.
    Hence, $\bigcup_{\mathcal{P} \in \text{GF}(p^2)} K_\mathcal{P}$ is a graph with the required properties.
\end{proof}

\begin{figure}[h]
    \begin{center}
        \begin{subfigure}{0.3\textwidth}
            \resizebox{\textwidth}{!}{
                \begin{tikzpicture}[thick,
                  vertex/.style={draw, circle, thick, fill=yellow, inner sep=0pt, minimum size=5pt},
                  edge/.style={draw=red, line width=2pt},
                  fuzzy/.style={line width=2pt, dashed}
                ]
                    \foreach \x in {1,...,5}
                    \foreach \y in {1,...,4}
                    {
                        \node[vertex] (\x-\y) at (\x, \y) {};
                    }
                    \draw[decoration={brace, raise=15pt}, decorate] (1-4) to node[above=16pt] {$p$} (5-4);
                    \draw[edge] (4-2) to (3-3);
                    \draw[fuzzy] (3-3) to (2-4);
                    \draw[fuzzy] (5-1) to (4-2);
                    \begin{pgfonlayer}{bg}
                        \foreach \y in {1,...,4}
                        {
                            \node[draw, rounded rectangle, dashed, fit=(1-\y) (5-\y), inner sep=7pt, label=right:$V_\y$] {};
                        }
                    \end{pgfonlayer}
                \end{tikzpicture}
            }
            \caption{No edge belongs to two different lines (cliques).}
        \end{subfigure}
        \begin{subfigure}{0.3\textwidth}
            \resizebox{\textwidth}{!}{
                \begin{tikzpicture}[thick,
                  vertex/.style={draw, circle, thick, fill=yellow, inner sep=0pt, minimum size=5pt},
                  edge/.style={line width=2pt}
                ]
                    \foreach \x in {1,...,5}
                    \foreach \y in {1,...,4}
                    {
                        \node[vertex] (\x-\y) at (\x, \y) {};
                    }
                    \draw[decoration={brace, raise=15pt}, decorate] (1-4) to node[above=16pt] {$p$} (5-4);
                    \draw[edge] (4-3) to node {\Huge\color{red} $\times$} (3-3);
                    \begin{pgfonlayer}{bg}
                        \foreach \y in {1,...,4}
                        {
                            \node[draw, rounded rectangle, dashed, fit=(1-\y) (5-\y), inner sep=7pt, label=right:$V_\y$] {};
                        }
                    \end{pgfonlayer}
                \end{tikzpicture}
            }
            \caption{No edges within layers, so no cliques larger than $2c$.}
        \end{subfigure}
        \caption{Clique gadget proof sketch for theorem~\ref{thm:explicit-vertex-anti}.}\label{fig:clique-gadget-proof}
    \end{center}
\end{figure}
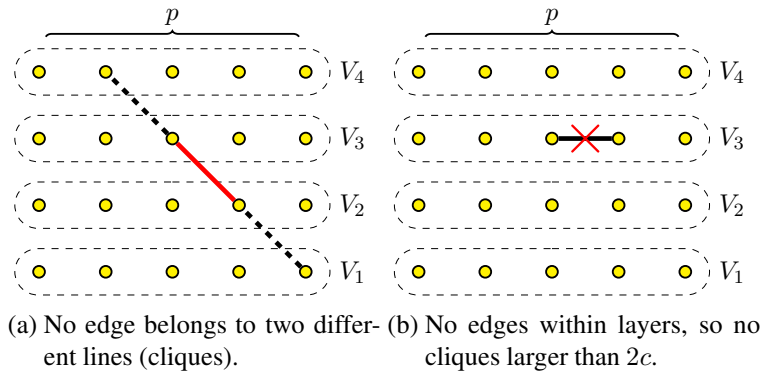

\begin{proof}[Proof of Theorem~\ref{thm:explicit-vertex-anti}]
    Suppose we have an algorithm $\mathcal{C}$ for explicit vertex streams which can, with probability at least $\frac{2}{3}$, produce a $c$-approximation to $\omega(G)$, the size of the largest clique. We will show that such an algorithm can be used to solve $\textsc{Chain}_{2c}$, by communicating its state $2c-1$ times.
    
Fix an instance of $\textsc{Chain}_{2c}$ with vectors of length $b = \frac{n^2}{64c^4}$.
Our lower bound implies any algorithm that can solve this must send at least one message of size $\Omega \left( \frac{b}{c^3} \right) = \Omega \left( \frac{n^2}{c^7} \right)$.
Take $n$ vertices and partition the nodes into $2c$ groups of size $\frac{n}{2c}$. Each group will be added to the stream by one of the parties.

\paragraph{Intra-party edges.}
First, consider the group of nodes associated with party $P_i$. We will encode the bits of $X^{(i)}$ onto the internal edges of this group using the construction from Lemma~\ref{lem:erasure}. The size $\frac{n}{c}$ sub-graph can fit $b$ cliques of size $2c$. We include the edges of clique $j$ if and only if $X^{(i)}_j = 1$. This is well-defined as the cliques are edge-disjoint. Label the clique in party $P_i$ corresponding to bit $j$ of $X^{(i)}$ as $\mathcal{K}^i_j$. The final party $P_{2c}$ has no associated vector. Instead, it constructs a single clique of size $2c$ and leaves the other vertices isolated.

\paragraph{Inter-party edges.}
We also need edges between the sub-graphs associated with
    different parties.
    Each party $P_i$ will connect all its vertices to some of the vertices belonging to previous parties ($P_j$ for $j< i$). These edges are considered to belong to party $P_i$, as they will be added by this party in the vertex streaming model.
    For each $j < i$ the party $P_i$ connects every one of its vertices to all of $\mathcal{K}^j_{\sigma_j}$ (the clique corresponding to index $\sigma_j$). For this to happen, $P_i$ must know all $\sigma_j$ for $j < i$. This information is not known initially, but can be appended to the communications between players with only $O(c)$ overhead.

\medskip    
    Now that we have our construction, we need to show bounds on $\omega(G)$ for the two cases. First, consider when every $X^{(i)}_{\sigma_i} = 1$. In this case we have each of the cliques $\mathcal{K}^i_{\sigma_i}$ present and connected together, forming a clique of size $4c^2$.  
    Now consider the case when every $X^{(i)}_{\sigma_i} = 0$. Consider a clique $\mathcal{K}$ in the graph. If $\mathcal{K}$ contains multiple vertices belonging to one party $P_i$, then it can contain none from any subsequent party $P_j$ ($j > i$), and at most one from each preceding party $P_l$ ($l < i$). Hence the size of any clique is bounded by $4c - 1$.
    To see why this holds, observe that for any $i < 2c$, our clique can contain only one vertex from $\mathcal{K}^i_{\sigma_i}$, as none of its edges are included in the graph. So to contain multiple vertices from party $P_i$, the clique $\mathcal{K}$ must contain a vertex $v$ from some $\mathcal{K}^i_j$ with $j \neq \sigma_i$. But then all subsequent parties $P_j$ ($j > i$) will have no vertices adjacent to $v$, so cannot contribute anything to $\mathcal{K}$. So the best we can do is include one vertex from each $\mathcal{K}^i_{\sigma_i}$ and then $2c$ from party $P_{2c}$ giving a clique of size $4c - 1$.
    
    To complete the proof, observe that this gap in clique sizes can be distinguished by a $c$-approximation algorithm, and any streaming algorithm gives a communication protocol by having each party update the algorithm state with their information and then passing it to the next player.
\end{proof}

\begin{figure}[h]
    \begin{center}
        \resizebox{\textwidth}{!}{
            \begin{tikzpicture}[
                vertex/.style={draw, circle, thick, fill=yellow, inner sep=0pt, minimum size=5pt},
                h/.style={draw=red, line width=2pt},
                o/.style={draw=orange, line width=1pt},
                g/.style={draw=black!40!green, line width=1pt},
                b/.style={draw=blue, line width=1pt},
                n/.style={draw=black, line width=2pt}
                ]
                \foreach \x in {1,...,5}
                \foreach \y in {1,...,4}
                \foreach \p in {1,...,3}
                {
                    \node[vertex] (\x-\y-\p) at (\x+6*\p, \y) {};
                }
                \foreach \y in {1,...,4}
                \foreach \p in {4}
                {
                    \node[vertex] (1-\y-\p) at (1+6*\p, \y) {};
                }
                \begin{pgfonlayer}{bg}
                    \foreach \y in {1,...,4}
                    \foreach \p in {1,...,3}
                    {
                        \node[draw, rounded rectangle, dashed, fit=(1-\y-\p) (5-\y-\p), inner sep=7pt] {};
                    }
                    \foreach \p in {1,...,3}
                    {
                        \node[draw, rounded corners=10pt, dashed, fit=(1-1-\p) (5-4-\p), inner sep=12pt, label=below:$P_\p$] (P-\p) {};
                    }
                    \foreach \p in {4}
                    {
                        \node[draw, rounded corners=10pt, dashed, fit=(1-1-\p) (1-4-\p), inner sep=12pt, label=below:$P_\p$] (P-\p) {};
                    }
                \end{pgfonlayer}
                \draw[h] (5-1-1) to (4-2-1);
                \draw[h] (4-2-1) to (3-3-1);
                \draw[h] (3-3-1) to (2-4-1);
                \draw[o] (5-1-1) to (P-2);
                \draw[o] (4-2-1) to (P-2);
                \draw[o] (3-3-1) to (P-2);
                \draw[o] (2-4-1) to (P-2);
                \draw[g, rounded corners=10pt] (P-3) |- (13, 5);
                \draw[g, in=180] (5-1-1) to (13, 5);
                \draw[g, in=180] (4-2-1) to (13, 5);
                \draw[g, in=180] (3-3-1) to (13, 5);
                \draw[g, in=180] (2-4-1) to (13, 5);
                \draw[b, rounded corners=10pt] (P-4) |- (11, 5.5);
                \draw[b, in=180] (5-1-1) to (11, 5.5);
                \draw[b, in=180] (4-2-1) to (11, 5.5);
                \draw[b, in=180] (3-3-1) to (11, 5.5);
                \draw[b, in=180] (2-4-1) to (11, 5.5);
                \draw[h] (3-1-2) to (3-2-2);
                \draw[h] (3-2-2) to (3-3-2);
                \draw[h] (3-3-2) to (3-4-2);
                \draw[g, out=0, in=180] (3-1-2) to (18, 5);
                \draw[g, out=0, in=180] (3-2-2) to (18, 5);
                \draw[g, out=0, in=180] (3-3-2) to (18, 5);
                \draw[g, out=0, in=180] (3-4-2) to (18, 5);
                \draw[b, in=180] (3-1-2) to (17, 5.5);
                \draw[b, in=180] (3-2-2) to (17, 5.5);
                \draw[b, in=180] (3-3-2) to (17, 5.5);
                \draw[b, in=180] (3-4-2) to (17, 5.5);
                \draw[h] (2-1-3) to (3-2-3);
                \draw[h] (3-2-3) to (4-3-3);
                \draw[h] (4-3-3) to (5-4-3);
                \draw[b, in=180] (2-1-3) to (24, 5.5);
                \draw[b, in=180] (3-2-3) to (24, 5.5);
                \draw[b, in=180] (4-3-3) to (24, 5.5);
                \draw[b, in=180] (5-4-3) to (24, 5.5);
                \draw[n] (1-1-4) to (1-2-4);
                \draw[n] (1-2-4) to (1-3-4);
                \draw[n] (1-3-4) to (1-4-4);
            \end{tikzpicture}
        }
        \caption{Example lower bound instance with $4$ players for theorem~\ref{thm:explicit-vertex-anti}. Cliques corresponding to $\sigma_1$, $\sigma_2$, and $\sigma_3$ are shown in bold red---other cliques are omitted.}\label{fig:clique-lower-bound}
    \end{center}
\end{figure}
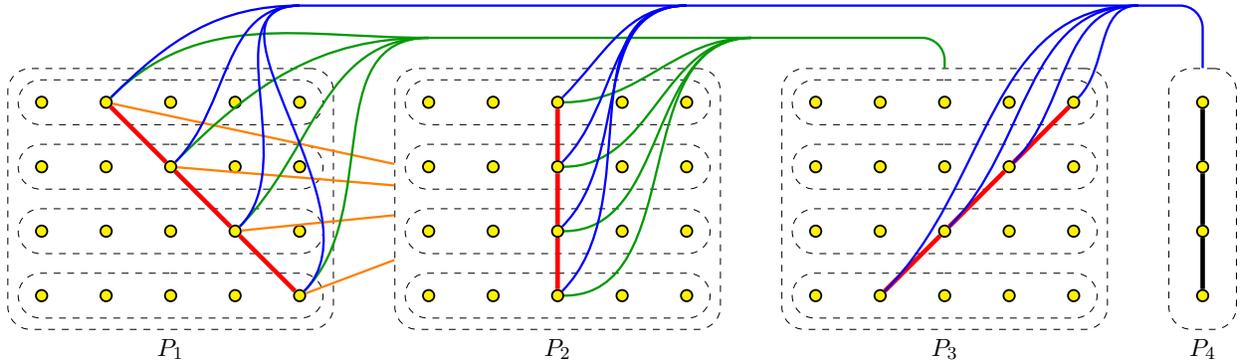

Interestingly, the same construction gives us hardness for approximating the chromatic number of a graph. This is notably not possible in the $2$-party edge stream construction in \cite{halldorsson2012streaming}, as the random graphs used as gadgets have large chromatic number with high probability (see \cite{bollobas1988chromatic}).

\begin{corollary}\label{thm:explicit-vertex-chromatic}
    Any explicit vertex streaming algorithm to find a $c$-approximation to $\chi(G)$ (the chromatic number), succeeding with probability at least $2/3$ requires $\Omega\left( \frac{n^2}{c^7} \right)$ space.
\end{corollary}

\begin{proof}
    Consider the construction in the proof of Theorem~\ref{thm:explicit-vertex}. In the case of all $X^{(i)}_{\sigma_i} = 1$, the graph contains a clique of size $4c^2$, so it requires at least as many colours.
    
    Conversely, in the case of every $X^{(i)}_{\sigma_i} = 0$, we can construct a $4c$-colouring of the graph. First colour each of the nodes in each $\mathcal{K}^i_{\sigma_i}$ with the $i$\textsuperscript{th} colour (this is allowed, as they have no internal edges). Now, the remaining vertices in each party are not adjacent to any uncoloured vertices from other parties, so we simply need to be able to complete the colouring of each party in isolation with $2c$ new colours and we are finished. This is easily done, as each party's sub-graph is $2c$-partite by construction.
\end{proof}

\section{Maximum Independent Set in Geometric Intersection Graphs}

In this section we present a collection of  results around geometric
intersection graphs given as explicit or implicit vertex streams.
This represents a first study of how the difficulty of these problems differs between the models, and with other factors such as dimension.

Recall that a geometric intersection graph is a graph where nodes correspond to geometric objects, and edges indicate whether or not a particular pair of objects intersect. These graphs can be described implicitly, by just giving the collection of geometric objects, or explicitly as a collection of vertices and edges under the promise that some geometric representation exists.
For sufficiently complex geometry, every graph will have a geometric representation (simply take hyper-rectangles of a high enough dimension), so to define meaningful classes, we must limit the universe of objects.

We will consider the $l^p$ closed balls in $\mathbb{R}^d$, particularly for $p = 1$, $2$, and $\infty$. A $d$-dimensional $l^p$ ball is uniquely defined by its center and radius. The ball with center $c$ and radius $r$ is exactly the region $\lbrace x \in \mathbb{R}^d \,:\, \|x - c\|_p \leq r \rbrace$. For the implicit representation, we discretize the space of possible centers as $[M]^d$ and the space of possible radii as $[M]$. Then, the input stream is defined as follows:

\begin{definition}
    A $d$-dimensional $l^p$ ball stream consists of sequence of $n$ pairs in $[M]^d \times [M]$ indicating the center and radius of each ball respectively.
    
    This defines an implicit vertex stream, with the intersection function:
    \[\sigma\left( (p_1, w_1), (p_2, w_2) \right) =
    \begin{cases}
        1, & \text{if } \| p_1 - p_2 \|_p \leq w_1 + w_2\\
        0, & \text{otherwise.}
    \end{cases}
    \]
\end{definition}

Typically $M$ will be some polynomial in $n$, so that the balls can be
represented in $\text{polylog}(n)$ space.
We also define a dilation parameter, which will capture some of the difficulty of the problem.

\begin{definition}
    Let the dilation, $\Delta = \frac{r_{\text{max}}}{r_{\text{min}}}$, be the ratio between the largest and smallest radii in the stream.
\end{definition} 

We refer to 
the interesting special case of $\Delta=1$, when all the balls are
uniform in size, as a unit ball stream.

\subsection{Interval Graphs: $\mathbf{d=1}$}

In $d=1$, the choice of $p$ is irrelevant, as any $l^p$ ball is simply an interval.
As discussed in Section~\ref{sec:prior-work}, given an interval (ball) stream we can compute a $\frac{3}{2}$-approximation to MIS in $\tilde{O}(\alpha(G))$ space, and any better approximation requires $\Omega(n)$ space.
A natural question is how this compares with the space complexity for an interval intersection graph given as an explicit vertex stream.

\begin{theorem}\label{thm:explicit-interval}
    Any algorithm with constant error probability that returns a $(\frac{5}{3} - \epsilon)$-approximation of $\alpha(G)$ for a unit interval intersection graph given as an explicit vertex stream requires $\Omega(n)$ space.
\end{theorem}

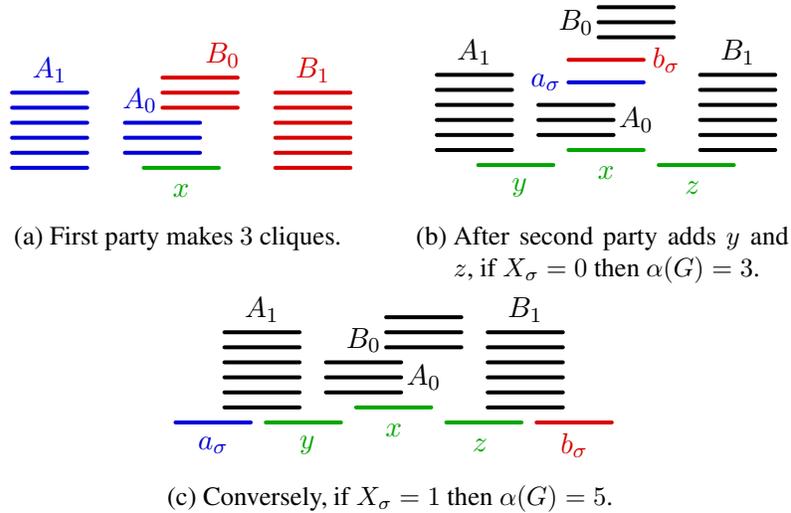
\begin{figure}[h]
    \centering
    \subcaptionbox{First party makes $3$ cliques.}[0.3\textwidth]{
        \begin{tikzpicture}[thick,
          edgeA/.style={blue!85!black, line cap=round, line width=1.5pt},
          edgeB/.style={red!85!black, line cap=round, line width=1.5pt},
          edgeC/.style={green!65!black, line cap=round, line width=1.5pt}
        ]
            \foreach \y in {1,...,6}
            {
                \draw[edgeA] (0,0.2*\y) to (1,0.2*\y);
                \draw[edgeB] (3.5,0.2*\y) to (4.5,0.2*\y);
            }
            \node[blue!85!black] (A-1) at (0.5, 1.5) {$A_1$};
            \node[red!85!black] (B-1) at (4, 1.5) {$B_1$};
            \foreach \y in {1,...,3}
            {
                \draw[edgeA] (1.5,0.2+0.2*\y) to (2.5,0.2+0.2*\y);
                \draw[edgeB] (2,0.8+0.2*\y) to (3,0.8+0.2*\y);
            }
            \node[blue!85!black] (A-0) at (1.7, 1.1) {$A_0$};
            \node[red!85!black] (B-0) at (2.8, 1.7) {$B_0$};
            \draw[edgeC] (1.75,0.2) to (2.75,0.2);
            \node[green!65!black] (x) at (2.25, -0.1) {$x$};
        \end{tikzpicture}
    }
    \hspace{0.5cm}
    \subcaptionbox{After second party adds $y$ and $z$, if $X_\sigma = 0$ then $\alpha(G)=3$.}[0.3\textwidth]{
        \begin{tikzpicture}[thick,
          edge/.style={line cap=round, line width=1.5pt},
          edgeA/.style={blue!85!black, line cap=round, line width=1.5pt},
          edgeB/.style={red!85!black, line cap=round, line width=1.5pt},
          edgeC/.style={green!65!black, line cap=round, line width=1.5pt}
        ]
            \foreach \y in {1,...,6}
            {
                \draw[edge] (0,0.2*\y) to (1,0.2*\y);
                \draw[edge] (3.5,0.2*\y) to (4.5,0.2*\y);
            }
            \node (A-1) at (0.5, 1.5) {$A_1$};
            \node (B-1) at (4, 1.5) {$B_1$};
            \foreach \y in {1,...,3}
            {
                \draw[edge] (1.35,0.2+0.2*\y) to (2.35,0.2+0.2*\y);
                \draw[edge] (2.15,1.5+0.2*\y) to (3.15,1.5+0.2*\y);
            }
            \node (A-0) at (2.65, 0.6) {$A_0$};
            \node (B-0) at (1.85, 1.9) {$B_0$};
            \draw[edgeC] (1.75,0.2) to (2.75,0.2);
            \draw[edgeC] (0.55,0) to (1.55,0);
            \draw[edgeC] (2.95,0) to (3.95,0);
            \node[green!65!black] (x) at (2.25, -0.1) {$x$};
            \node[green!65!black] (y) at (1.1, -0.3) {$y$};
            \node[green!65!black] (z) at (3.4, -0.3) {$z$};
            \draw[edgeA] (1.75,1.1) to (2.75,1.1);
            \draw[edgeB] (1.75,1.4) to (2.75,1.4);
            \node[blue!85!black] (a-s) at (1.45, 1.1) {$a_\sigma$};
            \node[red!85!black] (b-s) at (3.05, 1.4) {$b_\sigma$};
        \end{tikzpicture}
    }
    \subcaptionbox{Conversely, if $X_\sigma = 1$ then $\alpha(G)=5$.}[0.5\textwidth]{
        \begin{tikzpicture}[thick,
          edge/.style={line cap=round, line width=1.5pt},
          edgeA/.style={blue!85!black, line cap=round, line width=1.5pt},
          edgeB/.style={red!85!black, line cap=round, line width=1.5pt},
          edgeC/.style={green!65!black, line cap=round, line width=1.5pt}
        ]
            \foreach \y in {1,...,6}
            {
                \draw[edge] (0, 0.2*\y) to (1, 0.2*\y);
                \draw[edge] (3.5, 0.2*\y) to (4.5, 0.2*\y);
            }
            \node (A-1) at (0.5, 1.5) {$A_1$};
            \node (B-1) at (4, 1.5) {$B_1$};
            \foreach \y in {1,...,3}
            {
                \draw[edge] (1.35, 0.2+0.2*\y) to (2.35, 0.2+0.2*\y);
                \draw[edge] (2.15, 0.8+0.2*\y) to (3.15, 0.8+0.2*\y);
            }
            \node (A-0) at (2.65, 0.6) {$A_0$};
            \node (B-0) at (1.85, 1.1) {$B_0$};
            \draw[edgeC] (1.75, 0.2) to (2.75, 0.2);
            \draw[edgeC] (0.55, 0) to (1.55, 0);
            \draw[edgeC] (2.95, 0) to (3.95, 0);
            \node[green!65!black] (x) at (2.25, -0.1) {$x$};
            \node[green!65!black] (y) at (1.1, -0.3) {$y$};
            \node[green!65!black] (z) at (3.4, -0.3) {$z$};
            \draw[edgeA] (-0.65, 0) to (0.35, 0);
            \draw[edgeB] (4.15, 0) to (5.15, 0);
            \node[blue!85!black] (a-s) at (-0.15, -0.3) {$a_\sigma$};
            \node[red!85!black] (b-s) at (4.65, -0.3) {$b_\sigma$};
        \end{tikzpicture}
    }
    \caption{Interval representations for the construction in theorem~\ref{thm:explicit-interval}. Horizontal positioning represents the location of the intervals in $\mathbb{R}$, vertical positioning is for clarity only.}
    \label{fig:explicit-interval}
\end{figure}

\begin{proof}
    We will show this bound by a reduction from the $2$-party \textsc{Index} communication problem for an $n$-bit vector.

    Consider an instance of \textsc{Index} with bit vector $X \in \lbrace 0, 1 \rbrace^n$ and index to be queried $\sigma \in [n]$. We will construct a $2n+3$ vertex graph as an explicit vertex stream.

    Label the vertices $x$, $y$, $z$ and $a_i$, $b_i$ for $i \in [n]$. Split the $a_i$'s into two sets based on the bit vector $X$: $A_1 = \lbrace a_i \rbrace_{X_i = 1}$ and $A_0 = \lbrace a_i \rbrace_{X_i = 0}$. Similarly let $B_1 = \lbrace b_i \rbrace_{X_i = 1}$ and $B_0 = \lbrace b_i \rbrace_{X_i = 0}$.
    Now the first party creates the following subgraph in the stream: a clique consisting of all the vertices in $A_1$, a second clique made from $B_1$, and a third clique containing $A_0 \cup B_0 \cup \lbrace x \rbrace$.

    So far this represents a valid interval graph, which can be interpreted as three adjacent ``stacks'' of intervals. Now, the second player adds $y$ with edges to every $a_i$ except $a_\sigma$ and then adds $z$ with edges to every $b_i$ except $b_\sigma$. This can still be viewed as a valid interval graph, but we now require some intervals from each stack to be ``shifted'' to overlap with the two new intervals.

    In the case of $X_\sigma = 0$, the resulting graph has $\alpha(G) = 3$. Otherwise, $\alpha(G) = 5$. Hence, any algorithm giving a better than $\frac{5}{3}$ (one-sided) approximation factor could distinguish them and solve \textsc{Index}.
\end{proof}

This shows that MIS for interval graphs is strictly more difficult in explicit vertex streams than implicit ones.

\subsection{Square Graphs: $\mathbf{d=2}$}

Observe that in $d=2$, the geometry defined by $l^1$ and $l^\infty$
balls is equivalent, after fixed scaling and rotation.
Thus any instance in one geometry can be transformed into an
equivalent (intersection-wise) instance of the other, after ensuring
that the discretization is not too coarse.

Our first result for $2$ dimensions is a $3$-approximation algorithm
for MIS on a unit square stream.
This is a generalization the algorithm of \cite{cabello2017interval} for unit interval streams --- we perform a decomposition of the plane into $2$-by-$3$ strips, similar to their decomposition of the line into length $3$ segments.

\begin{theorem}\label{thm:unit-squares}
For $d=2$ and $\Delta=1$ we can $3$-approximate MIS for an $l_1$ or $l_\infty$ ball stream using $\tilde{O}(\alpha(G))$ space.
\end{theorem}

\begin{proof}
Let $r$ be the radius of the balls, and let $w = 2r$. Without loss of generality, we consider the $l_\infty$ version of the problem.
    
First we look at the problem restricted to squares in the half
open strip $[0,3w) \times [0,2w)$, referring to the first axis as
``left-right'', and the second as ``up-down''.
At most $2$ non-overlapping closed unit squares can fit fully within
the open region, 
and for them not to overlap one must be left of the other.
Hence, by storing the leftmost and rightmost squares seen within the strip, we can determine exactly an MIS in the region.
    
Now, consider the whole of $[M]^2$. We partition this up into
$3w$-by-$2w$ half-open strips and consider only the squares which
fall exactly within one of the strips.
By solving MIS within each strip as before and taking the union, we can solve exactly MIS on this substream. This requires us to store at most twice as many squares as the solution.
    
    Finally, we consider $6$ different copies of the partitioning shifted by $0$, $w$, or $2w$ horizontally and $0$ or $w$ vertically. Any square from the stream must be fully contained in a strip in at least $2$ of the $6$ partitionings. In particular, this holds for every square in a fixed MIS of the full stream. Using $G_{x,y}$ for the graph of the substream of squares that fit exactly in the partitioning shifted by $(x, y)$, we know the following: $\sum\limits_{x=0,w,2w} \left( \sum\limits_{y=0,w} \left( \alpha(G_{x,y}) \right)\right) \geq 2\alpha(G)$. That is, the sum of the sizes of the substream MIS's for the $6$ partitionings is at least $2$ times the size of the true MIS. Therefore $\max\limits_{x=0,w,2w} \left( \max\limits_{y=0,w} \left( \alpha(G_{x,y}) \right)\right) \geq \frac{1}{3}\alpha(G)$
    
    Therefore, we simply take the largest of the $6$ substream MIS's and this must give a $3$-approximation. In total, we must store at most $12$ squares per MIS square.
\end{proof}

As in \cite{cabello2017interval} for unit intervals, this immediately leads to a sublinear space algorithm for estimating $\alpha(G)$ with only a $(1+\epsilon)$ factor loss in approximation factor, through a combination of distinct elements and sampling.

\begin{corollary}
For $d=2$ and $\Delta=1$ we can $(3 + \epsilon)$-approximate $\alpha(G)$ with constant probability for an $l_1$ or $l_\infty$ ball stream using $O(\epsilon^{-2}\log{\epsilon^{-1}}+\log{n})$ space.
\end{corollary}
\begin{proof}
    Observe that if we can get a $(1+\epsilon)$-approximation to each $\alpha(G_{x,y})$ from the proof of theorem~\ref{thm:unit-squares}, then we are done.
    
    Each strip can have $0$, $1$, or $2$ disjoint squares in it. To approximate $\alpha(G_{x,y})$ we estimate $\gamma$, the number of strips of a given partitioning which are non-empty, and $\delta$, the average number of disjoint squares in the non-empty strips. Then $\alpha(G_{x,y}) \approx \gamma\delta$.
    
    Observe that approximating $\gamma$ is a distinct elements problem, which we can $(1+\epsilon)$-approximate with constant probability in $O(\epsilon^{-2} + \log{n})$ space (see \cite{kane2010optimal}).
    
    Then $\delta$ can be approximated by using nearly-uniform permutations to keep a nearly-uniform sample of the non-empty strips. For full details see \cite[lemma 16]{cabello2017interval}.
\end{proof}

One might speculate whether this decomposition approach could afford a
better approximation factor based on some different partioning of the
place.
We give evidence for the negative, since 
any larger strips result in the fixed-size sub-problems not being
solvable exactly, as the following result shows.

\begin{theorem}
	Given a stream of $w$-by-$w$ squares (2 dimensional $l_\infty$ balls) contained in a $(2 + \delta)w$-by-$(2 + \delta)w$ region, achieving a $(\frac{3}{2}-\epsilon)$-approximation to $\alpha(G)$, with constant probability of success, for any $\epsilon, \delta > 0$ requires $\Omega(n)$ space.
\end{theorem}

\begin{proof}
We show this by a reduction from $2$-party \textsc{Index}.
Fix an instance with bit vector $X \in \lbrace 0,1 \rbrace^n$ and query index $\sigma \in [n]$.
For the lower bound, we use squares of width $w = \frac{4n}{\delta}$, which meets the requirements for an implicit vertex stream as long as $\delta$ is constant.
    
Party one constructs the following collection of squares arranged
along a diagonal line: for each $i \in [n]$ with $X_i = 1$ include the square centered on $(\frac{2n}{\delta} + 2i, \frac{2n}{\delta} + 2n + 2 - 2i)$.
Now, observe that the parties could use a $(\frac{3}{2} - \epsilon)$-approximation streaming algorithm to allow party two to determine $X_\sigma$. Simply have party one run the algorithm on its collection, then pass the state to party two and append squares centered on $(\frac{6n}{\delta} + 2\sigma + 1, \frac{2n}{\delta} + 2n + 2 - 2\sigma)$ and $(\frac{2n}{\delta} + 2\sigma, \frac{6n}{\delta} + 2n + 3 - 2\sigma)$. If $X_\sigma = 1$, there exists a square in the original collection sandwiched between the two new squares giving $\alpha(G) = 3$. Otherwise, $\alpha(G) = 2$.
\end{proof}

Our next result for two dimensions is a stronger lower bound for approximating $\alpha(G)$ of a stream of unit squares in an unrestricted region, based on a reduction from the chained index communication problem used in our main result in Section~\ref{sec:general-explicit}.

\begin{theorem}\label{thm:unit-squares-lb}
Achieving a $(\frac{5}{2}-\epsilon)$-approximation of $\alpha(G)$, with constant probability of success, on an $l_1$ or $l_\infty$ ball stream with $d=2$ and $\Delta=1$ requires $\Omega(n)$ space for any $\epsilon > 0$.
\end{theorem}

\begin{proof}
We show the lower bound for $l_\infty$, which then implies the result for $l_1$. 
This proof works by reducing from the $3$-party chained index
problem.
As with the general result in Section~\ref{sec:general-explicit}, more than $2$ parties are necessary in order to give a bound for approximation factors greater than $2$.
    
    Suppose we have an instance of $\textsc{Chain}_3$ with $n$-bit vectors. We will describe a way for each party to construct a collection of unit $l_1$ balls based on the part of the input they hold, such that their union has small or large $\alpha(G)$ depending on the solution to the communication problem. Thus, a streaming algorithm for approximating $\alpha(G)$ can be used to solve the communication problem, giving a space bound.
    
For the lower construction we use domain size $M = 10n^3$ and ball radius $r = 2n^2$, which are small enough to allow succinct $\operatorname{polylog}n$ sized descriptions of the balls (as required for an implicit vertex stream), but large enough to create the gadgets we require.
Fix integer $k \in [n]$. The first party has the bit vector $X^{(1)}$.
For each entry with $X^{(1)}_i = 1$, add balls centered at $(i(4n+3)+(j+1)(4n^2+3n), 4n^2)$ and $(i(4n+3)+(j+1)(4n^2+3n), 8n^2)$ for each $j \in [k]$.
Essentially, this makes two horizontal lines of balls stacked on top
of each other.
There are potentially $nk$ ball locations along the line with centers $4n+3$ apart.
The first $n$ locations are associated with the $n$ entries of
$X^{(1)}$; we place a ball if $X^{(1)}=1$, and omit it otherwise; then
this is repeated $k$ times in succession. 
The two lines produce a collection of balls $G_1$ of size at most
$2nk$, with $\alpha(G_1) = 2k$.
Importantly, the collection of balls associated with any index $X^{(1)}_i = 1$ forms an MIS.
    
Party two will obliviously add their own set of balls, such that $\alpha(G)$ will increase exactly when the answer bit $X^{(1)}_{\sigma_1} = 1$. The second party has its bit vector $X^{(2)}$ and the index $\sigma_1$ of the answer bit in the first party's bit vector. For each entry with $X^{(2)}_i = 1$ add a ball centered at $(\sigma_1(4n+3)+(j+\frac{3}{2})(4n^2+3n), 6n^2-n+2i)$ for each $j \in [k]$.
Essentially, this produces $k$ columns of $n$ balls. The columns are
spaced with centers $n^2+3n$ apart, lined up to fit in the gaps
between the balls corresponding to bit $X^{(1)}_{\sigma_1}$ of the
first party (if those balls are present).
Each column is constructed as follows: place $n$ balls spaced with
centers $1$ apart, each associated with one of the entries of
$X^{(2)}$, but exclude each ball whose $X^{(2)}_i = 0$.
The result is a collection of balls $G_2$ of size $nk$ with
$\alpha(G_2) = k$.
Again, the collection of balls associated with any index $X^{(1)}_i = 1$ forms an MIS.
    
Now, party three will add the final collection of balls $G_3$. Party
three knows $\sigma_1$ (this can be appended to any message from player two) and $\sigma_2$, the index of the answer bit in
party two's bit vector.
This party adds balls centered at
$(\sigma_1(4n+3)+(j+\frac{3}{2})(4n^2+3n), 10n^2-n+2\sigma_2 + 1)$ and
$(\sigma_1(4n+3)+(j+\frac{3}{2})(4n^2+3n), 2n^2-n+2\sigma_2 - 1)$ for
each $j \in [k]$.
These balls sit at the top and bottom of each column from party $2$
sandwiching the ball corresponding to $X^{(2)}_{\sigma_2}$ (if
present).
This final collection has an independent set of size $2k$.
    
At this point we wish to determine $\alpha(G)$ for the union $G = G_1 \cup G_2 \cup G_3$. In the case that $X^{(1)}_{\sigma_1} = X^{(2)}_{\sigma_2} = 1$, we can take the MIS of $G_1$ associated with $X^{(1)}_{\sigma_1}$ and the MIS of $G_2$ associated with $X^{(2)}_{\sigma_2}$ along with all the balls in $G_3$ to form an MIS of $G$ of size $5k$.
However, if $X^{(1)}_{\sigma_1} = X^{(2)}_{\sigma_2} = 0$, then
choosing any ball from $G_1$ (other than the left-most
$\sigma_1-1$ balls in each row) excludes every ball in a column of
$G_2$ and a ball from $G_3$.
Similar exclusions occur between the other pairs of collections.
The result is that the best we can do is to choose an MIS from $G_1$ corresponding to an index smaller than $\sigma_1$ along with $1$ ball from each of $G_2$ and $G_3$ in the rightmost column, giving a total of $\alpha(G) = 2k+2$.
    
This shows that a streaming algorithm achieving an approximation
factor better than $\frac{5k}{2k+2}$ must use $\Omega(n)$ space.
This holds for any constant $k$ (just take $n$ large enough to allow that $k$), giving the result.
\end{proof}

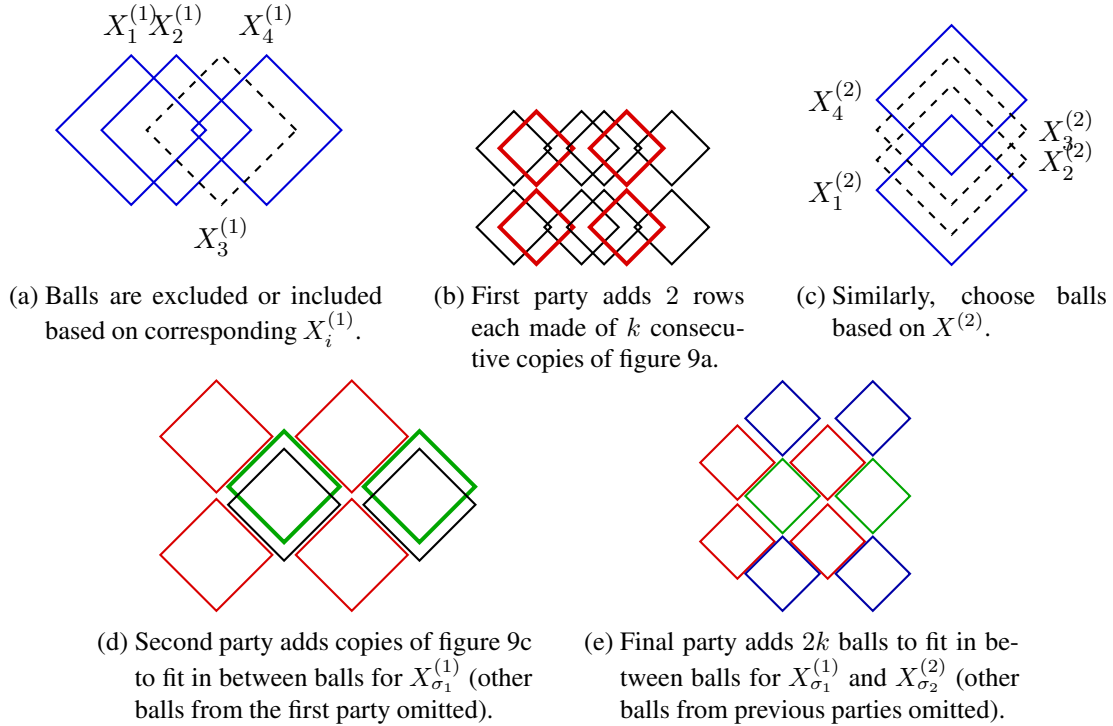
\begin{figure}[h]
    \centering
    \subcaptionbox{Balls are excluded or included based on corresponding $X^{(1)}_i$.\label{subfig-squares-gadget-one}}[0.3\textwidth]{
        \begin{tikzpicture}[thick,
          ball/.style={draw, diamond, inner sep=0pt, minimum height=2cm, minimum width=2cm}
        ]
            \foreach \x in {1,2,4}
            {
                \node[ball, blue!85!black, label=above:{$X^{(1)}_\x$}] (1-1-\x) at (0.6*\x, 0) {};
            }
            \foreach \x in {3}
            {
                \node[ball, dashed, label=below:{$X^{(1)}_\x$}] (1-1-\x) at (0.6*\x, 0) {};
            } 
        \end{tikzpicture}
    }
    \hspace{0.5cm}
    \subcaptionbox{First party adds $2$ rows each made of $k$ consecutive copies of figure~\ref{subfig-squares-gadget-one}.}[0.25\textwidth]{
        \begin{tikzpicture}[thick,
          ball/.style={draw, diamond, inner sep=0pt, minimum height=1cm, minimum width=1cm}
        ]
            \foreach \z in {1,2}
            \foreach \y in {1,2}
            {
                \foreach \x in {1,4}
                {
                    \node[ball] (\z-\y-\x) at (1.2*\y+0.3*\x, \z*1.05) {};
                }
                \foreach \x in {2}
                {
                    \node[ball, red!85!black, line width=1.5pt] (\z-\y-\x) at (1.2*\y+0.3*\x, \z*1.05) {};
                }
            }
        \end{tikzpicture}
    }
    \hspace{0.5cm}
    \subcaptionbox{Similarly, choose balls based on $X^{(2)}$.\label{subfig-squares-gadget-two}}[0.25\textwidth]{
        \begin{tikzpicture}[thick,
          ball/.style={draw, diamond, inner sep=0pt, minimum height=2cm, minimum width=2cm}
        ]
            \foreach \x in {1,4}
            {
                \node[ball, blue!85!black, label=left:{$X^{(2)}_\x$}] (1-1-\x) at (0, 0.4*\x) {};
            }
            \foreach \x in {2,3}
            {
                \node[ball, dashed, label=right:{$X^{(2)}_\x$}] (1-1-\x) at (0, 0.4*\x) {};
            } 
        \end{tikzpicture}
    }
    \subcaptionbox{Second party adds copies of figure~\ref{subfig-squares-gadget-two} to fit in between balls for $X^{(1)}_{\sigma_1}$ (other balls from the first party omitted).}[0.35\textwidth]{
        \begin{tikzpicture}[thick,
          ball/.style={draw, diamond, inner sep=0pt, minimum height=1.5cm, minimum width=1.5cm}
        ]
            \foreach \z in {1,2}
            \foreach \y in {1,2}
            {
                \node[ball, red!85!black] (\z-\y-3) at (1.8*\y, \z*1.575) {};
            }
            \foreach \y in {1,2}
            {
                \foreach \x in {4}
                {
                    \node[ball, green!65!black, line width=1.5pt] () at (0.9+1.8*\y, 2.1625+0.08*\x) {};
                }
                \foreach \x in {1}
                {
                    \node[ball] () at (0.9+1.8*\y, 2.1625+0.08*\x) {};
                }
            }
        \end{tikzpicture}
    }
    \hspace{0.5cm}
    \subcaptionbox{Final party adds $2k$ balls to fit in between balls for $X^{(1)}_{\sigma_1}$ and $X^{(2)}_{\sigma_2}$ (other balls from previous parties omitted).}[0.35\textwidth]{
        \begin{tikzpicture}[thick,
          ball/.style={draw, diamond, inner sep=0pt, minimum height=1cm, minimum width=1cm}
        ]
            \foreach \z in {1,2}
            \foreach \y in {1,2}
            {
                \node[ball, red!85!black] (\z-\y-3) at (1.2*\y, \z*1.05) {};
            }
            \foreach \y in {1,2}
            \foreach \x in {4}
            {
                \node[ball, green!65!black] () at (0.6+1.2*\y, 1.44166667+0.05333333333*\x) {};
                \node[ball, blue!65!black] () at (0.6+1.2*\y, 2.475+0.05333333333*\x) {};
                \node[ball, blue!65!black] () at (0.6+1.2*\y, 0.40833333333+0.05333333333*\x) {};
            }
        \end{tikzpicture}
    }
    \caption{Example for theorem~\ref{thm:unit-squares-lb} with $X^{(1)} = (1, 1, 0, 1)$, $X^{(2)} = (1, 0, 0, 1)$, $\sigma_1 = 2$, $\sigma_2 = 4$, and $k=2$.}
    \label{fig:unit-squares-lb}
\end{figure}

If we are allowed a combination of large and small balls, we can slightly improve the lower bound up to the maximum possible for a $3$-party construction.

\begin{theorem}
Achieving a $(3-\epsilon)$-approximation of $\alpha(G)$, with constant probability of success, on an $l_1$ or $l_\infty$ ball stream with $d=2$ and arbitrary $\Delta$ requires $\Omega(n)$ space for any $\epsilon > 0$.
\end{theorem}

\begin{proof}
We adapt the construction from Theorem~\ref{thm:unit-squares-lb} as follows: the first party inserts $k$ rows stacked on top of each other, rather than $2$, the second party inserts copies of its columns between every consecutive pair of rows from the first party, and the third party places smaller balls in between the columns of party $2$ such that they are independent of the squares corresponding to the answer bit but overlap the squares either side.
    
This results in a gap of $k^2+k+1$ to $3k^2$ for the two cases, giving the result.
\end{proof}

\section{Conclusion}

We have addressed the complexity of Maximal and Maximum Independent Set (and
various relaxations and related problems) under three natural models of graph streams:
edge-arrival, explicit vertex arrival, and implict vertex arrival.
The central problems of maximal and maximum independent set separate
the models: maximal independent set has high space complexity in
the edge-arrival model but an easy greedy algorithm in the vertex
models; however, MIS has high space complexity for explicit vertex
arrival (and hence also the harder edge arrival model), while there
are small-space constant factor approximations for some cases of implicit
vertex arrival. 

A natural extension is to consider {\em weighted} versions of the MIS
problem, where each vertex is assigned a weight and the aim is to
maximize the weight of the independent set.
Our upper bounds generalize naturally to this case.
Consider, for example, the two-dimensional geometric case ($l_1$ or
$l_\infty$).
A naive first approach is to round all weights to the nearest power of
epsilon, and solve each weight class separately, before combining the
results.
An improved result follows by observing that for each ``strip'' in the
decomposition of Theorem~\ref{thm:unit-squares-lb}, we only need to
track at most $\tilde{O}(\epsilon^{-1})$ weight classes: if a given strip
contains a ball of weight $w$ then we can ignore any intervals of
weight less than $\epsilon w$ without affecting the quality of the
solution by more than a $(1+\epsilon)$ factor.
Thus we obtain a $(3+\epsilon)$-approximate weighted MIS for the unit
square case using only $\tilde{O}(\log W +
\frac{\alpha(G)}{\epsilon})$. 
It remains to consider the complexity of other weighted variants of
MIS problems. 

There are a number of other natual open questions that follow from our
study: 
\begin{itemize}
    \item Is there a multi-pass lower bound for \emph{maximal} independent set in edge streams?
    \item Are there $o(\alpha(G))$ space algorithms for achieving constant factor approximations to $\alpha(G)$ for classes of geometric intersection graphs given as \emph{explicit} vertex streams?
    \item Can we close the gap between the upper and lower bounds for approximating MIS in a unit square stream?
    \item Is there a constant factor approximation algorithm for MIS on streams of arbitrary sized squares?
\end{itemize}

\printbibliography

\end{document}